\documentclass[journal]{IEEEtran}
\IEEEoverridecommandlockouts
\usepackage{cite}
\usepackage[hidelinks]{hyperref}
\usepackage{amsmath,amssymb,amsfonts}
\usepackage{algorithmic}
\usepackage{graphicx}
\usepackage{textcomp}
\usepackage{xcolor}
\usepackage{CJK}
\usepackage{algorithm}
\usepackage{pifont}
\usepackage{amsthm}

\newtheorem{thm}{\bf Theorem}
\usepackage{subfig}
\usepackage{float}  

\usepackage{booktabs} 

\usepackage{algorithmic}
\usepackage{amsmath}
\usepackage{amssymb}
\usepackage{psfrag}
\usepackage{graphicx}
\usepackage{epstopdf}
\usepackage{multirow}
\usepackage{stfloats}
\usepackage{color}
\usepackage[normalem]{ulem}
\usepackage{arydshln}
\usepackage{enumerate}
\usepackage{bbm}\usepackage{verbatim}

\begin{document}

\title{Hierarchical Cognitive Spectrum Sharing in Space-Air-Ground Integrated Networks\\
}

\author{
Zizhen Zhou, Qianqian Zhang, Jungang Ge, 
and Ying-Chang Liang, \IEEEmembership{Fellow, IEEE}
\thanks{
Z. Zhou, Q. Zhang, and J. Ge are with the National Key Laboratory of Wireless Communications, University of Electronic Science and Technology of China (UESTC), Chengdu 611731, China (e-mails: zhouzizhen@std.uestc.edu.cn; qqzhang\_kite@163.com; and gejungang@std.uestc.edu.cn).

Y.-C. Liang is with the Center for Intelligent Networking and Communications (CINC), University of Electronic Science and Technology of China (UESTC), Chengdu 611731, China (e-mail: liangyc@ieee.org).
}
}

\maketitle

\begin{abstract}
In space-air-ground integrated networks (SAGINs), cognitive spectrum sharing has been regarded as a promising solution to meet the rapidly increasing spectrum demand of various applications, because it can significantly improve the spectrum efficiency by enabling a secondary network to access the spectrum of a primary network. 
However, different networks in SAGIN may have different quality of service (QoS) requirements, which can not be well satisfied with the traditional cognitive spectrum sharing architecture. 
For example, the aerial network typically has high QoS requirements, which however may not be met when it acts as a secondary network.
To address this issue, in this paper, we propose a hierarchical cognitive spectrum sharing architecture (HCSSA) for SAGINs, where the secondary networks are divided into a preferential one and an ordinary one.
Specifically, the aerial and terrestrial networks can access the spectrum of the satellite network under the condition that the caused interference to the satellite terminal is below a certain threshold. 
Besides, considering that the aerial network has a higher priority than the terrestrial network, we aim to use a rate constraint to ensure the performance of the aerial network.
Subject to these two constraints, we consider a sum-rate maximization for the terrestrial network by jointly optimizing the transmit beamforming vectors of the aerial and terrestrial base stations. 
To solve this non-convex problem, we propose a penalty-based iterative beamforming (PIBF) scheme that uses the penalty method and the successive convex approximation technique.
Moreover, we also develop three low-complexity schemes, where the beamforming vectors are obtained by optimizing the normalized beamforming vectors and power control. 
Finally, we provide extensive numerical simulations to compare the performance of the proposed PIBF scheme and the low-complexity schemes.
The results also demonstrate the advantages of the proposed HCSSA compared with the traditional cognitive spectrum sharing architecture.

\end{abstract}


\section{Introduction}
\label{sec_introduction}
In recent years, space-air-ground integrated networks (SAGINs) have been extensively studied to provide diverse communication services \cite{liu2018space}. 
The terrestrial network, such as the the fifth-generation (5G) wireless communication network, can provide high data rate services, but with limited coverage. 
The satellite networks, which consist of satellites and terrestrial infrastructures, can provide seamless connectivity to remote areas, but with a relatively low transmission rate.
The aerial network, which mainly includes a variety of aircrafts, such as unmanned aerial vehicles (UAV), airships, balloons, airplanes, and so on, is also thriving due to its advantage of flexible deployment \cite{zeng2019accessing}. 
SAGIN, formed by the integration of these three complementary networks, is designed to accommodate a variety of services and applications with diverse communication requirements in various scenarios. 
It has received considerable attention from both academia and industry. 
However, the rapid emergence of diverse communication requirements makes it challenging to allocate limited spectrum resources to meet the demands.
Dynamic spectrum management (DSM) is regarded as a promising solution as it provides an effective and flexible spectrum allocation method. In particular, the key enabling technology of
DSM is cognitive radio, which allows the secondary network to use the spectrum of the primary network \cite{liang2011cognitive}. 
When a primary network is active, its spectrum can be utilized by the secondary user under the condition that the caused interference to the primary user is below a certain threshold, a.k.a., the interference temperature limit \cite{liang2011cognitive}.

Cognitive spectrum sharing has been extensively studied in terrestrial networks \cite{yuan2020intelligent}, satellite-terrestrial networks \cite{liang2021realizing, sharma2013transmit, lin2019joint}, and aerial-terrestrial networks, such as UAV communications \cite{saleem2015integration} and air-to-ground (A2G) communications \cite{wang2010cognitive, jacob2016cognitive, zhang2014aeronautical}.
In \cite{sharma2013transmit}, the signal-to-interference-plus-noise ratio (SINR) of the secondary network is maximized under the interference temperature constraint of the primary user by optimizing the beamforming of the terrestrial base station (BS).
In \cite{lin2019joint}, under the quality of service (QoS) constraints of satellite and cellular users, the beamforming and power allocation of the satellite and terrestrial networks are optimized to maximize the sum rate of the satellite-terrestrial network.
%
%
With the emergence of new applications, such as aircraft passenger communications, the need for high transmission rates for A2G communications is increasing, where airplanes as aerial users are served by dedicated BSs on the ground \cite{lin20215g, mozaffari2021toward}.
In \cite{zhang2014aeronautical}, a central cognitive structure for broadband A2G communications is studied, in which the airplanes that transmit important data are treated as primary users and others as secondary users.
In SAGIN, since airplanes suffer from interference from numerous terrestrial BSs and A2G BSs also interfere with other networks when A2G communications use the same spectrum as the satellite and terrestrial networks, interference control is needed to ensure their performance \cite{tadayon2016inflight, lin2021sky}.

%
In SAGIN, cognitive spectrum sharing has also been extensively studied \cite{liu2020cell, lei2021joint, guo2022multi, wang2019dynamic, liu2023resource, hua2019joint, pervez2021joint}.
To expand the coverage of the terrestrial network and make up for the shortcomings of the satellite network, such as limited rate and large delay, UAV and the high altitude platform (HAP) are used as aerial BSs.
%
In \cite{liu2020cell, lei2021joint, guo2022multi}, the cognitive satellite-UAV network is studied, where multiple UAVs as aerial BSs serve their users by sharing the spectrum of the satellite network.
Data transmission efficiency maximization, sum secrecy rate maximization, and transmission delay minimization in the secondary network are focused by \cite{liu2020cell}, \cite{lei2021joint}, and \cite{guo2022multi}, respectively.
%
%
In \cite{wang2019dynamic, liu2023resource}, the cognitive satellite-HAP network is studied, where the satellite network acts as the primary network and the HAP network acts as the secondary network to serve their respective ground users.
A sum-rate maximization problem for the secondary network is formulated, where the channel allocation and the power allocation of a HAP are optimized.
%
In \cite{hua2019joint}, coordinated multi-point transmission in cognitive satellite-terrestrial networks assisted by a UAV is studied, in which the UAV and a terrestrial BS cooperate to serve terrestrial users in the secondary network by using the spectrum of the satellite network.
The BS/UAV transmission power allocation and UAV trajectory are jointly optimized to maximize the sum rate of secondary users.
%
In \cite{pervez2021joint}, a cognitive SAGIN is studied in which the satellite network is the primary network, while the vehicles on the ground are served by UAVs and a terrestrial BS in the secondary network.
The vehicle association, BS/UAV transmission power allocation, and UAV trajectory are jointly optimized to maximize the throughput of all vehicles.

In SAGIN, different networks in SAGIN are heterogeneous and have diverse needs.
However, most of the above studies adopt the traditional cognitive spectrum sharing architecture (TCSSA), which can not well satisfy the diverse needs.
For example, the airplane often has a high QoS requirement to support flight control as well as some emerging applications, such as in-flight entertainment and communications \cite{baltaci2021survey}. 
%
In such cases, the QoS requirement may not be met in TCSSA since the performance of aerial users, who suffer from severe interference from terrestrial BSs, is not protected.
Keeping this in mind, the spectrum sharing structure of the Citizen Broadband Radio Service (CBRS) can be a promising solution \cite{sohul2015spectrum, ying2018sas}.
The CBRS includes three service tiers, i.e., incumbent access (IA), priority access (PA), and general authorized access (GAA).
The incumbents have the highest priority and are protected from interference caused by the PA and GAA tiers.
PA users have QoS requirements and are authorized to operate with certain interference protection.
GAA devices have the lowest priority, are not protected from interference, and should be avoided from interfering with incumbents and PA users.
%
In particular, CBRS allows GAA to use unallocated or unused PA channels on an opportunistic basis \cite{sohul2015spectrum}.
Different from CBRS, in this paper, we focus on concurrent spectrum access, which allows each tier of the network to transmit on the same frequency band if the high-priority tier is well protected from the co-channel interference.
 

Motivated by the above considerations, to satisfy the diverse needs of different networks, we propose a hierarchical cognitive spectrum sharing architecture (HCSSA) for SAGIN, in which the satellite network acts as the primary network, the aerial network acts as the preferential secondary network with a QoS requirement, and the terrestrial network acts as the ordinary secondary network.
%
Specifically, the satellite network shares its spectrum with the aerial and terrestrial networks if the caused interference to the satellite terminal is below a threshold.
%
%
Under the interference temperature constraint and the rate constraint of the aerial user, we aim to maximize the sum rate of the terrestrial network by jointly optimizing the beamforming vectors of the aerial and terrestrial networks.
%
Since the formulated problem is non-convex, we first transform it into a convex one by introducing auxiliary variables and performing first-order Taylor expansion.
Then, we propose a penalty-based iterative beamforming (PIBF) scheme by exploiting the penalty-based method and the successive convex approximation (SCA) technique.
In order to improve the convergence, we propose an initialization algorithm for the PIBF scheme.
Nevertheless, due to the large number of optimization variables, the complexity of the PIBF scheme is quite high.
%
To overcome this, we also develop three alternative low-complexity beamforming schemes, where the beamforming vectors are obtained by optimizing the normalized beamforming vectors and power control, and the number of variables in the optimization problem is reduced.
%
The main contributions of the paper are summarized as follows:
\begin{itemize}
\item To satisfy the diverse needs of different networks, we propose a hierarchical cognitive spectrum sharing architecture for SAGIN, where the aerial network acts as a secondary network and is protected by a QoS constraint.
The proposed architecture can not only meet the communication requirements of the aerial network preferentially under the condition of limited resources but also further exploit the underutilized resources of the aerial network to improve the performance of the terrestrial network.
\item We formulate the beamforming design problem to maximize the sum rate of the terrestrial network under the interference temperature constraint of the satellite terminal and the rate constraint of the aerial user.
We propose a PIBF scheme and three low-complexity schemes to solve the formulated non-convex problem.
\item Finally, simulation results show that the PIBF scheme can obtain a higher sum rate of the terrestrial network than the three low-complexity schemes.
Moreover, the results also
demonstrate that when the aerial user has a high rate requirement, HCSSA can preferentially meet this requirement while TCSSA cannot.
\end{itemize}

The rest of this paper is organized as follows:
Section \ref{secSystemModel} introduces the system model of SAGIN considered in this paper and formulates the hierarchical cognitive spectrum sharing problem.
In Section \ref{sec_Proposed_Scheme}, a PIBF scheme is proposed to solve the formulated problem.
Section \ref{sec_Low_Com_Scheme} provides three low-complexity suboptimal schemes.
Section \ref{sec_Simulation_Results} shows the simulation results.
Finally, Section \ref{sec_Conclusions} concludes this paper.


Notations used in this paper are listed as follows.
The lowercase, bold lowercase, and bold uppercase, i.e., $a$, ${\bf{a}}$, and ${\bf{A}}$ are scalar, vector, and matrix, respectively.
$\mathbb{C}^{a \times b}$ denotes the space of $a \times b$ complex-valued matrices.
${\bf{I}}_a$ denotes an identity matrix of size $a \times a$.
$| \cdot |$ denotes the absolute value. 
${{\left\| {\bf{a}} \right\|}_2}$ denotes the $\ell_2$ norm of vector ${\bf{a}}$.
$(\cdot )^T$, $(\cdot )^H$, and ${\rm{Tr}}(\cdot )$ denote transpose, conjugate transpose, and trace, respectively.
${\cal C}{\cal N} (\mu, \sigma^2)$ denote the complex Gaussian distribution with mean $\mu$ and variance $\sigma^2$.

\section{System Model and Problem Formulation}
\label{secSystemModel}


\begin{figure}[t]
\centering
\includegraphics[width=8cm]{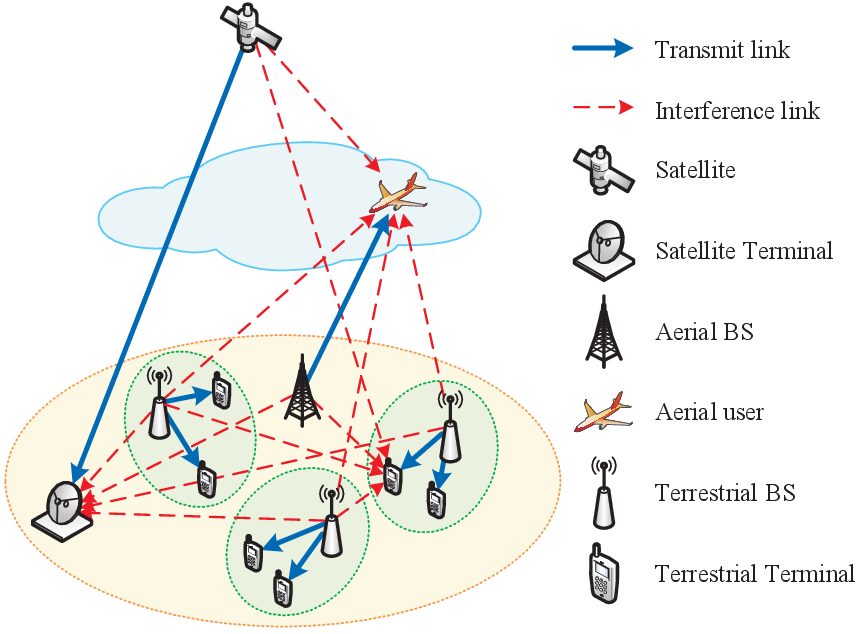}
\caption{Considered SAGIN, where the satellite network, the aerial network, and the terrestrial network all transmit downlink signals.}
\label{scenario}
\vspace{-0.4cm}
\end{figure}


%

As shown in Fig. \ref{scenario}, we consider a cognitive SAGIN consisting of three-tier networks, namely, a satellite network, an aerial network, and a terrestrial network.
The aerial network acts as the preferential secondary network with a QoS requirement and the terrestrial network acts as the ordinary secondary network.
Both of them are able to access the spectrum of the satellite network.
Moreover, the interference from the two secondary networks should satisfy the interference temperature constraint of the primary satellite network.
%

Specifically, in the satellite network, a geosynchronous Earth orbit (GEO) satellite with $M_S$ antennas serves a ground satellite terminal with a single antenna. 
In the aerial network, an aerial BS on the ground with $M_A$ antennas serves an aerial user (such as a UAV, an airplane, and so on) with a single antenna\footnote{We mainly focus on illustrating the advantages of HCSSA, so assume that the aerial user is relatively stationary.
For cases where the aerial user is dynamic, the proposed architecture still works, but the beamforming vectors need to be updated periodically, which is beyond the scope of this paper.}. 
The terrestrial network is a multicell multi-user spatial multiplex system with $N$ cells.
In each cell, there is a terrestrial BS with $M_G$ antennas and $K_n$ ($n=1,\dots, N$) terrestrial terminals with a single antenna.
%
%
%


In the following, we will introduce the channel model and the transmission model for the terrestrial terminals and the aerial user.
Then, we formulate a hierarchical cognitive spectrum sharing problem to jointly design the transmit beamforming vectors of the terrestrial and aerial BSs.

\subsection{Channel Model}

The channels from the terrestrial BS in the $n$-th ($n=1,\dots, N$) cell to the satellite terminal, the aerial user, and the $l$-th ($l=1,\dots, K_m$) receiver in the $m$-th cell are denoted by ${\bf{h}}_{S,n} \in \mathbb{C}^{M_G \times 1}$, ${\bf{h}}_{A,n} \in \mathbb{C}^{M_G \times 1}$, and ${\bf{h}}_{n,m,l} \in \mathbb{C}^{M_G \times 1}$, respectively.
The channels from the aerial BS to the satellite terminal, the aerial user, and the $k$-th receiver in the $n$-th cell are denoted by ${\bf{g}}_{S} \in \mathbb{C}^{M_A \times 1}$, ${\bf{g}}_{A} \in \mathbb{C}^{M_A \times 1}$, and ${\bf{g}}_{n,k} \in \mathbb{C}^{M_A \times 1}$, respectively.
The channels from the satellite to the aerial user and the $k$-th receiver in the $n$-th cell are denoted by ${\bf{f}}_{A} \in \mathbb{C}^{M_S \times 1}$ and ${\bf{f}}_{n,k} \in \mathbb{C}^{M_S \times 1}$, respectively.



As the aerial user is usually located high, the channels from the terrestrial and aerial BSs to the aerial user are usually dominated by the line-of-sight (LoS) component.
Thus, these channels are assumed to follow Rician fading, i.e.,
\begin{align}
\label{Rician_channel}
{\overline {\bf{h}}} = \sqrt {{L_{\rm{LoS}}^{-1}}\left( {d,f} \right)} \left( {\sqrt {\frac{\kappa }{{1 + \kappa }}} {\bf{h}}^{\rm{LoS}} + \sqrt {\frac{1}{{1 + \kappa }}} {\bf{h}}^{\rm{NLoS}} } \right),
\end{align}
with ${\overline {\bf{h}}} \in \{{\bf{h}}_{A,n}, {\bf{g}}_{A}\}$, 
where ${L_{\rm{LoS}}}\left( {d,f} \right)$, $\kappa$, ${\bf{h}}^{\rm{LoS}}$, and ${\bf{h}}^{\rm{NLoS}}$ denote the LoS path-loss, the Rician factor, the LoS components, and the non-line-of-sight (NLoS) components, respectively.
The LoS path-loss in ${\rm{dB}}$ are given by ${L_{\rm{LoS}}}\left( {d,f} \right) = 28 + 22{\log _{10}}\left( d \right) + 20{\log _{10}}\left( f \right)$, where $d$ is the distance between transmitter and receiver in ${\rm{m}}$ and $f$ is the carrier frequency in ${\rm{GHz}}$ \cite{access2010further}. 
The elements of the NLoS components ${\bf{h}}^{\rm{NLoS}}$ follow the Gaussian distribution, i.e., ${\cal C}{\cal N} (0, 1)$.

To illustrate the composition of ${\bf{h}}^{\rm{LoS}}$, we first define the steering vector as ${{\bf {a}}_{X}}\left(\vartheta\right) = \left[{ {1,{e^{j\frac {2\pi \overline d}{\lambda }\sin \vartheta }},\cdots,{e^{j\frac {2\pi \overline d}{\lambda }\left({{X - 1} }\right)\sin \vartheta }}} }\right]^{T}$, where
$\lambda$ is the carrier wavelength, 
$\overline d$ is the antenna element separation, 
$X$ is the dimension of the vector, and
$\vartheta $ is the angular parameter, i.e., the angle of departure (AoD) or angle of arrival (AoA).
%
Then, the LoS components of ${\bf{h}}_{A,n}$ and ${\bf{g}}_{A}$ can be expressed as ${\bf{h}}^{\rm{LoS}} = {{\bf {a}}_{M_G}}\left(\vartheta_{A,n}\right)$ and ${\bf{h}}^{\rm{LoS}} = {{\bf {a}}_{M_A}}\left(\vartheta_{AA}\right)$, respectively, where $\vartheta_{A,n}$ and $\vartheta_{AA}$ is the AoA or AoD of the signal from the terrestrial BS in the $n$-th cell and the aerial BS to the aerial user, respectively.
Considering that the channels from the terrestrial and aerial BSs to the receiver on the ground are usually dominated by the NLoS components, these channels are assumed to follow Rayleigh fading, i.e.,
\begin{align}
\label{Rayleigh_channel}
\widetilde {\bf{h}} = \sqrt {{L_{\rm{NLoS}}^{-1}}\left( {d,f} \right)} {\bf{h}}^{\rm{NLoS}},
\end{align}
with $\widetilde {\bf{h}} \in \{{\bf{h}}_{S,n}, {\bf{h}}_{n,m,l}, {\bf{g}}_{S}, {\bf{g}}_{n,k}\}$, 
where the NLoS path-losses in ${\rm{dB}}$ are given by ${L_{\rm{NLoS}}}\left( {d,f} \right) = 22.7 + 36.7{\log _{10}}\left( d \right) + 26{\log _{10}}\left( f \right)$ \cite{access2010further}. 



The channels from the satellite to the aerial user and the terrestrial terminals can be modeled as \cite{li2017resource}: 
\begin{align}
\label{Shadowed_Rician_channel}
{\overline {\bf{f}}} = \sqrt {{L_{\rm{LoS}}^{-1}}\left( {d,f} \right) b\left( \varphi  \right)} \widetilde {\bf{f}},
\end{align}
with ${\overline {\bf{f}}} \in \{{\bf{f}}_{A}, {\bf{f}}_{n,k}\}$, where $b\left( \varphi  \right)$ is the beam gain of the satellite and
$\widetilde {\bf{f}}$ is the channel fading vector that follows shadowed-Rician (SR) distribution \cite{abdi2003new, an2016secure}.
The SR distribution is parameterized by $(\Omega, b, m)$, where $\Omega$ is the average power of the LoS component, $2b$ is the average power of the scattering component, and $m$ is the Nakagami-$m$ parameter relating to the fading severity.
%
The beam gain of the satellite can be expressed as \cite{zheng2012generic, li2017resource}: 
\begin{align}
\label{Beam_Gain}
b\left( \varphi  \right) = {b_{\max}}{\left( {\frac{{{J_1}\left( u \right)}}{{2u}} + 36\frac{{{J_3}\left( u \right)}}{{{u^3}}}} \right)^2},
\end{align}
where $u = 2.07123\sin \varphi /\sin \left( {{\varphi _{{\rm{3dB}}}}} \right)$, 
$\varphi$ is the angle between the location of the receiver and the beam center with respect to the satellite, 
$\varphi _{\rm{3dB}}$ is the 3-dB angle, 
$b_{\max}$ is the maximal satellite antenna gain \cite{li2018robust}, 
$J_1(\cdot)$ and $J_3(\cdot)$ represent the first-kind Bessel function of order 1 and 3, respectively.




\subsection{Transmission Model}







\subsubsection{Terrestrial Network}
Denote the information signal for the $k$-th receiver in the $n$-th cell, the aerial user, and the satellite terminal by $x_{n,k}$, $x_{A}$, and $x_S$, respectively.
Besides, the corresponding beamforming vectors are denoted by ${{\bf{w}}_{n,k}} \in \mathbb{C}^{M_G \times 1}$, ${\bf{v}} \in \mathbb{C}^{M_A \times 1}$, and ${\bf{u}} \in \mathbb{C}^{M_S \times 1}$, respectively.
The received signal at the $k$-th receiver in the $n$-th cell can be written as:
\begin{align}
{y_{n,k}} = &\underbrace {{\bf{h}}_{n,n,k}^H{{\bf{w}}_{n,k}}{x_{n,k}}}_{{\rm{intended\; signal}}} + \underbrace {\sum_{(m,l) \ne (n,k)} {{\bf{h}}_{m,n,k}^H{{\bf{w}}_{m,l}}{x_{m,l}}} }
_{{\rm{intracell\; and\; intercell\; interference}}}\nonumber \\
&+ \underbrace {{\bf{g}}_{n,k}^H{\bf{v}}{x_A}}_{{\rm{aerial\; interference}}} + \underbrace {{\bf{f}}_{n,k}^H{{\bf{u}}}{x_S}}_{{\rm{satellite\; interference}}} + {z_{n,k}}, \nonumber
\end{align}
where $z_{n,k}\sim {\cal C}{\cal N} (0, \sigma _{n,k}^2)$ is the additive white Gaussian noise.
Then, the SINR of the $k$-th receiver in the $n$-th cell for decoding $x_{n,k}$ can be expressed as: 
\begin{align}
\label{SINR_kth_user_in_nth_cell}
{\gamma _{n,k}} = \frac{{{{\left| {{\bf{h}}_{n,n,k}^H{{\bf{w}}_{n,k}}} \right|}^2}}}{{\overline \sigma _{n,k}^2 + {{\left| {{\bf{g}}_{n,k}^H{\bf{v}}} \right|}^2} + \sum_{(m,l) \ne (n,k)} {{{\left| {{\bf{h}}_{m,n,k}^H{{\bf{w}}_{m,l}}} \right|}^2}} }},
\end{align}
where ${\overline \sigma }^2_{n,k} = \sigma _{n,k}^2 + {\left| {\bf{f}}_{n,k}^H{\bf{u}} \right|^2}$ represents the sum of the satellite interference power and the noise power and is known to the receiver.
Therefore, the achievable rate of the $k$-th receiver in the $n$-th cell can be expressed as: 
\begin{align}
\label{rate_kth_user_in_nth_cell}
{R_{n,k}} = {\log _2}\left( {1 + {\gamma _{n,k}}} \right).
\end{align}

\subsubsection{Aerial Network}
The received signal at the aerial user can be expressed as:  
\begin{align}
{y_A} = \underbrace {{\bf{g}}_{A}^H{\bf{v}}{x_A}}_{{\rm{intended\; signal}}} + \underbrace {\sum_{(m,l)} {{\bf{h}}_{A,m}^H{{\bf{w}}_{m,l}}{x_{m,l}}} }_{{\rm{terrestrial\; interference}}} 
+ \underbrace {{\bf{f}}_{A}^H{{\bf{u}}}{x_S}}_{{\rm{satellite\atop interference}}} + {z_A}, \nonumber
\end{align}
where $z_{A}\sim {\cal C}{\cal N} (0, \sigma _A^2)$ is the additive white Gaussian noise.
Then, the SINR of the aerial user for decoding $x_{A}$ can be expressed as:  
\begin{equation}
\beta  = \frac{{{{\left| {{\bf{g}}_A^H{\bf{v}}} \right|}^2}}}{{\overline \sigma _A^2 + \sum_{(m,l)} {{{\left| {{\bf{h}}_{A,m}^H{{\bf{w}}_{m,l}}} \right|}^2}} }},
\end{equation}
where ${\overline \sigma }^2_A = \sigma _A^2 + {\left| {\bf{f}}_{A}^H{{\bf{u}}} \right|^2}$ represents the sum of the satellite interference power and the noise power.
Therefore, the achievable rate of the aerial user can be expressed as: 
\begin{align}
{R_A} = {\log _2}\left( 1 + \beta \right).
\end{align}


\subsection{Problem Formulation}
Due to the heterogeneity of different networks, their requirements may be different.
The satellite network requires the received interference from the aerial and the terrestrial networks to be below a threshold.
The aerial user requires its rate to be higher than a threshold due to its strict QoS requirement.
Inspired by these, we formulate a hierarchical cognitive spectrum sharing problem.
Specifically, subject to the interference temperature constraint of the satellite terminal and the rate constraint of the aerial user, we aim to maximize the sum rate of the terrestrial network by jointly optimizing the transmit beamforming vectors of the terrestrial and aerial BSs, which is given by:  
\begin{subequations}
\label{Problem_sate_active}
\begin{align}
\label{Problem_sate_active_sum_rate_terr}
&\mathop{\max} _{ \left\{{{\bf{w}}_{n,k}}\right\}, {\bf{v}} } \;
\sum_{(n,k)} {R_{n,k}} \\
&\;\;\;\;{\rm{s}}.{\rm{t}}.\;\;\;\;\;
\label{cons_interference_power_sate}
\sum_{(n,k)}{{\left| {\bf{h}}_{S,n}^H {\bf{w}}_{n,k} \right|}^2}  + {\left| {\bf{g}}_{S}^H {\bf{v}} \right|^2} \le {\overline I _S},\\
\label{cons_transmit_power_terrestrial_BS}
&\;\;\;\;\;\;\;\;\;\;\;\;\;\;\sum_{k = 1}^{{K_n}} {{\bf{w}}_{n,k}^H{{\bf{w}}_{n,k}}}  \le {\overline p _n}, \forall n = 1,...,N,\\
\label{cons_transmit_power_aerial_BS}
&\;\;\;\;\;\;\;\;\;\;\;\;\;\;{\bf{v}}^H{\bf{v}} \le {\overline q},\\
\label{cons_SINR_aerial_user}
&\;\;\;\;\;\;\;\;\;\;\;\;\;\;R_A \ge {\overline R_A},
\end{align}
\end{subequations}
where ${\overline I _S}$ is the interference temperature of the satellite terminal, 
${\overline p _n}$ is the maximum transmit power of the terrestrial BS in the $n$-th cell, 
${\overline q}$ is the maximum transmit power of the aerial BS, and
${\overline R_A}$ is the minimum rate requirement for the aerial user.
%


The objective function \eqref{Problem_sate_active_sum_rate_terr} is in the form of a sum of logarithms of fractions, where the variables exist in both the numerator and the denominator, and thus \eqref{Problem_sate_active_sum_rate_terr} is non-convex.
Therefore, problem \eqref{Problem_sate_active} is non-convex and difficult to be solved directly.



\section{Penalty-Based Iterative Scheme}
\label{sec_Proposed_Scheme}

In this section, in order to deal with the non-convex problem \eqref{Problem_sate_active}, we first transform the original problem into a convex one.
Then, we propose a PIBF scheme to solve the transformed problem iteratively.
Further, to search for feasible initial points, an initialization algorithm is proposed.
Finally, the convergence performance and computational complexity of the PIBF scheme are analyzed.

\subsection{Problem Transformation}

Since problem \eqref{Problem_sate_active} is non-convex, we aim to transform it into a convex problem which is easy to handle.
Specifically, we first transform problem \eqref{Problem_sate_active} into an equivalent form by representing each variable in a corresponding matrix form, i.e., ${{\bf{W}}_{n,k}} = {{\bf{w}}_{n,k}}{\bf{w}}_{n,k}^H, \forall n,k$ and ${\bf{V}} = {\bf{v}}{\bf{v}}^H$, which needs to satisfy ${{\bf{W}}_{n,k}} \succeq 0, \forall n,k$, ${\bf{V}} \succeq 0$, ${\rm{rank}}\left( {{\bf{W}}_{n,k}} \right) = 1, \forall n,k$, and ${\rm{rank}}\left( {\bf{V}} \right) = 1$. 
The channels are also expressed in the corresponding matrix form: ${{\bf{H}}_{m,n,k}} = {{\bf{h}}_{m,n,k}}{\bf{h}}_{m,n,k}^H$, ${{\bf{H}}_{A,n}} = {{\bf{h}}_{A,n}}{\bf{h}}_{A,n}^H$, ${{\bf{H}}_{S,n}} = {{\bf{h}}_{S,n}}{\bf{h}}_{S,n}^H$, ${{\bf{G}}_{n,k}} = {{\bf{g}}_{n,k}}{\bf{g}}_{n,k}^H$, ${{\bf{G}}_{A}} = {{\bf{g}}_{A}}{\bf{g}}_{A}^H$, and ${{\bf{G}}_{S}} = {{\bf{g}}_{S}}{\bf{g}}_{S}^H$.

For ease of notation, we denote the collection of all ${\bf{W}}_{n,k}$ by ${\bf{W}}$ and define $s_{n,k}\left( {\bf{V}}, {\bf{W}} \right) \triangleq  \ln \! \left({{\overline \sigma }_{n,k}^2} \!+ {\rm{Tr}}\left( {{\bf{G}}_{n,k}}{\bf{V}} \right) \!+ \sum_{(m,l)} {{\rm{Tr}}\!\left( {{\bf{H}}_{m,n,k}}{{\bf{W}}_{m,l}} \right)}\!\right)$, which is a concave function.
By defining ${\alpha_{n,k}}\left( {\bf{V}}, {\bf{W}} \right) \triangleq \overline \sigma _{n,k}^2  + {\rm{Tr}}\left( {{{\bf{G}}_{n,k}}{\bf{V}}} \right) + \sum_{(m,l) \ne (n,k)} {{\rm{Tr}}\left( {{{\bf{H}}_{m,n,k}}{{\bf{W}}_{m,l}}} \right)}$, we have ${\log _2} \left( {1 + {\gamma _{n,k}}} \right) = {\log _2}(e) \left(s_{n,k} \left( {\bf{V}}, {\bf{W}} \right) -  \ln \left( {\alpha_{n,k}} \left( {\bf{V}}, {\bf{W}} \right) \right)\right)$.
Then, problem \eqref{Problem_sate_active} can be transformed into the following equivalent problem:
\begin{subequations}
\label{Problem_CVX_SDR}
\begin{align}
\label{Problem_CVX_objective}
&\mathop {\max}_{{\bf{V}}, \left\{ {{\bf{W}}_{n,k}}\right\} } \;
{\log _2}(e) \! \! \sum_{(n,k)} \!\!  {\left( {{s_{n,k}}\left( {{\bf{V}},\!{\bf{W}}} \right) \! - \! \ln \left( {{\alpha_{n,k}}\!\left( {{\bf{V}},\!{\bf{W}}} \right)} \right)} \right)}\\
&\;\;\;\;\;{\rm{s}}.{\rm{t}}.\;\;\;\;
\label{Problem_CVX_cons_inter_temp}
\sum_{(n,k)} {{\rm{Tr}}\left( {{{\bf{H}}_{S,n}}{{\bf{W}}_{n,k}}} \right)} + {\rm{Tr}}\left( {{{\bf{G}}_{S}}{\bf{V}}} \right) \le {{\overline I}_S},\\
\label{Problem_CVX_cons_terres_BS}
&\;\;\;\;\;\;\;\;\;\;\;\;\;\;\sum_{k = 1}^{{K_n}} {{\rm{Tr}}\left( {\bf{W}}_{n,k} \right)}  \le {{\overline p}_n},\forall n,\\
\label{Problem_CVX_cons_aerial_BS}
&\;\;\;\;\;\;\;\;\;\;\;\;\;\;{\rm{Tr}}\left( {\bf{V}} \right) \le {\overline q},\\
\label{Problem_CVX_cons_SINR}
&\;\;\;\;\;\;\;\;\;\;\;\;\;\;{\overline \beta} \! \left( \! {\overline \sigma }^2_A \! + \! \sum_{(m,l)} \!{{\rm{Tr}}\!\left( {{{\bf{H}}_{A,m}}{{\bf{W}}_{m,l}}} \right)} \!\! \right) \!\! \le \! {\rm{Tr}} \! \left( {{\bf{G}}_{A}}{\bf{V}} \right)\!,\\
\label{Problem_CVX_cons_semidefinite}
&\;\;\;\;\;\;\;\;\;\;\;\;\;\;{{\bf{W}}_{n,k}} \succeq 0, \forall n,k,{\bf{V}} \succeq 0,\\
\label{Problem_CVX_cons_rank1}
&\;\;\;\;\;\;\;\;\;\;\;\;\;\;{\rm{rank}}\left( {{\bf{W}}_{n,k}} \right) = 1, \forall n,k, {\rm{rank}}\left( {\bf{V}} \right) = 1,
\end{align}
\end{subequations}
where ${\overline \beta}$ is the SINR corresponding to ${\overline R_A}$, i.e., ${\overline R_A} = {\log _2}\left( {1 + {\overline \beta}} \right)$.


Since $\sum_{(n,k)} \ln \left( {\alpha_{n,k}}\left( {\bf{V}}, {\bf{W}} \right) \right)$ is concave, \eqref{Problem_CVX_objective} is non-concave.
To deal with this, we aim to maximize \eqref{Problem_CVX_objective} by maximizing its concave lower bound.
To find a concave lower bound of \eqref{Problem_CVX_objective}, we intend to replace $\ln \left( {\alpha_{n,k}}\left( {\bf{V}}, {\bf{W}} \right) \right)$ with its upper bound.
%
Then, we introduce auxiliary variables ${u_{n,k}}, \forall n,k$ to scale $\ln \left( {\alpha_{n,k}}\left( {\bf{V}}, {\bf{W}} \right) \right)$ up to $u_{n,k}$, which is given by:
\begin{align}
\label{Problem_CVX_auxi_vari_v}
{\alpha_{n,k}}\left( {\bf{V}}, {\bf{W}} \right) \le {e^{u_{n,k}}}, \forall n,k.
\end{align}
%
%
Consequently, we obtain a lower bound of \eqref{rate_kth_user_in_nth_cell}, i.e., ${{{\log }_2}\left( {1 + {\gamma _{n,k}}} \right)}  \ge {\log _2}(e)  {\left( {{s_{n,k}\left( {\bf{V}}, {\bf{W}} \right)} - {u_{n,k}}} \right)}$.
%
Then, we obtain the lower bound of \eqref{Problem_CVX_objective}:
\begin{align}
\label{Problem_objective_lower_bound}
\sum_{(n,k)} \!{{{\log }_2}\!\left( 1 + {\gamma _{n,k}} \right)}  \!\ge {\log _2}(e)\!\! \sum_{(n,k)}\! {\left( {{s_{n,k}\left( {\bf{V}}, {\bf{W}} \right)}\! -\! {u_{n,k}}} \right)}.
\end{align}
Since $s_{n,k}\left( {\bf{V}}, {\bf{W}} \right)$ is concave, the right side of \eqref{Problem_objective_lower_bound} is concave.
%
Thus, problem \eqref{Problem_CVX_SDR} is transformed into:
\begin{subequations}
\label{Problem_CVX_SDR_auxi_vari}
\begin{align}
\label{Problem_CVX_auxi_vari_objective}
&\mathop {\max}_{{\bf{V}}, \left\{ {{\bf{W}}_{n,k}}, {u_{n,k}} \right\}} \;
{\log _2}(e) \sum_{(n,k)} \left( {{s_{n,k}\left( {\bf{V}}, {\bf{W}} \right)} - {u_{n,k}}} \right) \\
&\;\;\;\;\;\;\;\;{\rm{s}}.{\rm{t}}.\;\;\;\;\;\;\;\;\;
\eqref{Problem_CVX_cons_inter_temp} - \eqref{Problem_CVX_cons_rank1}, \eqref{Problem_CVX_auxi_vari_v}. \nonumber
\end{align}
\end{subequations}
Note that when the optimal solution of problem \eqref{Problem_CVX_SDR_auxi_vari} is obtained, \eqref{Problem_CVX_auxi_vari_objective} is equal to \eqref{Problem_CVX_objective}. 
This is because if the left side of \eqref{Problem_CVX_auxi_vari_v} is less than the right side, then $u_{n,k}$ will decrease until both sides are equal.
If equalities in \eqref{Problem_CVX_auxi_vari_v} hold, then the equality in \eqref{Problem_objective_lower_bound} holds.
Then, \eqref{Problem_CVX_auxi_vari_objective} is equal to \eqref{Problem_CVX_objective}.
Besides, ${\log }_2(e)$ in the objective function is omitted for brevity in the following since it does not affect the solution of problem \eqref{Problem_CVX_SDR_auxi_vari}.

Nonetheless, the constraint \eqref{Problem_CVX_auxi_vari_v} is still non-convex due to the right side of inequality \eqref{Problem_CVX_auxi_vari_v} is an exponential function.
Besides, the rank-one constraint \eqref{Problem_CVX_cons_rank1} is also non-convex.
In the following, we will deal with the non-convex constraints \eqref{Problem_CVX_auxi_vari_v} and \eqref{Problem_CVX_cons_rank1}.

To deal with the non-convex constraint \eqref{Problem_CVX_auxi_vari_v}, we apply the first-order Taylor expansion to ${e^{u_{n,k}}}, \forall n,k$ around ${u_{n,k}^{(t)}}, \forall n,k$ in the $t$ iteration and obtain:
\begin{align}
\label{Problem_CVX_auxi_vari_v_Taylor1}
{e^{u_{n,k}^{(t)}}}\left( {{u_{n,k}} - u_{n,k}^{(t)} + 1} \right) \le {e^{u_{n,k}}}, \forall n,k.
\end{align}
%
By replacing ${e^{u_{n,k}}}$ with ${e^{u_{n,k}^{(t)}}}\left( {{u_{n,k}} - u_{n,k}^{(t)} + 1} \right)$ in constraint \eqref{Problem_CVX_auxi_vari_v}, we obtain a tighter convex constraint:
\begin{align}
\label{Problem_CVX_auxi_vari_v_Taylor}
{\alpha_{n,k}}\left( {\bf{V}}, {\bf{W}} \right) \le {e^{u_{n,k}^{(t)}}}\left( {{u_{n,k}} - u_{n,k}^{(t)} + 1} \right), \forall n,k.
\end{align}
%
By replacing constraint \eqref{Problem_CVX_auxi_vari_v} with \eqref{Problem_CVX_auxi_vari_v_Taylor}, problem \eqref{Problem_CVX_SDR_auxi_vari} is transformed into:
\begin{subequations}
\label{Problem_CVX_transform}
\begin{align}
    \label{Problem_CVX_transform_obj}
&\mathop {\max}_{{\bf{V}}, \left\{ {{\bf{W}}_{n,k}}, {u_{n,k}} \right\}} \;
\sum_{(n,k)} \left( {{s_{n,k}\left( {\bf{V}}, {\bf{W}} \right)} - {u_{n,k}}} \right) \\
&\;\;\;\;\;\;\;\;{\rm{s}}.{\rm{t}}.\;\;\;\;\;\;\;\;\;
\eqref{Problem_CVX_cons_inter_temp} - \eqref{Problem_CVX_cons_rank1}, \eqref{Problem_CVX_auxi_vari_v_Taylor}. \nonumber
\end{align}
\end{subequations}


Then, we deal with the non-convex constraint \eqref{Problem_CVX_cons_rank1}.
Specifically, since ${{\bf{W}}_{n,k}}, \forall n,k$ and ${\bf{V}}$ are positive semidefinite hermite matrices, we have ${\text{Tr}}\left( {\bf{W}}_{n,k} \right) \geq {\eta}\left( {\bf{W}}_{n,k} \right), \forall n,k$ and ${\text{Tr}}\left( {\bf{V}} \right) \geq {\eta}\left( {\bf{V}} \right)$, where ${\eta}( \cdot )$ denotes the largest eigenvalue of a matrix.
The equality ${\text{Tr}}\left( {\bf{X}} \right) = {\eta}\left( {\bf{X}} \right)$ holds if and only if ${\bf{X}}$ is a rank-one matrix.
Thus, the rank-one constraint \eqref{Problem_CVX_cons_rank1} can be rewritten as the following equivalent equality constraints: 
\begin{align}
\label{cons_rank1_equi}
{\text{Tr}}\left( {\bf{W}}_{n,k} \right) = {\eta}\left( {\bf{W}}_{n,k} \right),\forall n,k, {\text{and}}\;
{\text{Tr}}\left( {\bf{V}} \right) = {\eta}\left( {\bf{V}}\right).
\end{align}

To deal with equation constraints in \eqref{cons_rank1_equi}, we aim to use the penalty method \cite{ben1997penalty} to solve problem \eqref{Problem_CVX_transform}, where \eqref{cons_rank1_equi} is relaxed into a penalty term to be minimized.

In order to satisfy constraint \eqref{cons_rank1_equi} as much as possible, we aim to measure the violation of constraint \eqref{cons_rank1_equi} by defining a penalty term, which is given by $F \triangleq \sum_{(n,k)} {\left( {\text{Tr}}\left( {{{\bf{W}}_{n,k}}} \right) - {\eta}\left( {{{\bf{W}}_{n,k}}} \right) \right)}  + {\text{Tr}}\left( {\bf{V}} \right) - {\eta}\left( {\bf{V}} \right)$.
%
%
To minimize the penalty term $F$ and maximize \eqref{Problem_CVX_transform_obj},
we add $F$ to the objective function \eqref{Problem_CVX_transform_obj} and obtain the following optimization problem:
\begin{subequations}
\label{Problem_CVX_penalty_function1}
\begin{align}
\label{Problem_CVX_penalty_obj1}
&\mathop {\max}_{{\bf{V}}, \left\{ {{\bf{W}}_{n,k}}, {u_{n,k}} \right\}} \;
\sum_{(n,k)} \left( {{s_{n,k}\left( {\bf{V}}, {\bf{W}} \right)} - {u_{n,k}}} \right) - \xi F\\
&\;\;\;\;\;\;\;\;{\rm{s}}.{\rm{t}}.\;\;\;\;\;\;\;\;\;
\eqref{Problem_CVX_cons_inter_temp} - \eqref{Problem_CVX_cons_semidefinite}, \eqref{Problem_CVX_auxi_vari_v_Taylor}, \nonumber
\end{align}
\end{subequations}
where $\xi > 0$ is the penalty factor.
If $\xi$ is larger, then the penalty for violating the rank-one constraint \eqref{Problem_CVX_cons_rank1} is more severe.
As $\xi$ approaches positive infinity, the solution of problem \eqref{Problem_CVX_penalty_function1} always satisfies constraint \eqref{Problem_CVX_cons_rank1}.
Thus, problem \eqref{Problem_CVX_penalty_function1} is equivalent to problem \eqref{Problem_CVX_transform} \cite{ben1997penalty}.
However, when $\xi$ is large, the objective function \eqref{Problem_CVX_penalty_obj1} is mainly affected by the penalty term, while the effect of the sum-rate term on \eqref{Problem_CVX_penalty_obj1} is negligible.
To avoid this, we first set $\xi$ to a small value, and then increase $\xi$ as the number of iterations increases until constraint \eqref{Problem_CVX_cons_rank1} is satisfied enough.

Note that since the penalty term $F$ in \eqref{Problem_CVX_penalty_obj1} is non-convex, problem \eqref{Problem_CVX_penalty_function1} is still non-convex. 
Therefore, we employ SCA \cite{dinh2010local} to obtain a suboptimal solution of problem \eqref{Problem_CVX_penalty_function1}.
To maximize \eqref{Problem_CVX_penalty_obj1}, we intend to maximize its concave lower bound.
To find a concave lower bound of \eqref{Problem_CVX_penalty_obj1}, we aim to obtain the convex upper bound of the penalty term $F$.
%
Specifically, by performing the first-order Taylor expansion in the $t$ iteration on ${\eta}\left( {\bf{V}} \right)$ and ${\eta}\left( {\bf{W}}_{n,k} \right), \forall n, k$ at the given point ${\bf{V}}^{(t)}$ and ${\bf{W}}_{n,k}^{(t)}, \forall n, k$, respectively, we have:
\begin{align}
\label{penalty_lower_bound}
{\eta}\left( \!{{{\bf{W}}_{n,k}}}\! \right) \! \geq \! {\overline \eta}\left( \!{{\bf{W}}_{n,k}};\!{\bf{W}}_{n,k}^{(t)}\! \right)\!, \forall n,\!k, {\text{and}}\;
{\eta}\left( \!{\bf{V}} \!\right) \! \geq \! {\overline \eta}\left( \!{\bf{V}};\!{{\bf{V}}^{(t)}} \!\right)\!, 
\end{align}
where ${\overline \eta} \left( {\bf{X}}; {{\bf{X}}^{(t)}} \right) \triangleq {\text{Tr}}\left( {\boldsymbol{\theta}} \left( {{\bf{X}}^{(t)}} \right){{\left( {{\boldsymbol{\theta}} \left( {{\bf{X}}^{(t)}} \right)} \right)}^H} \left( {{\bf{X}} - {\bf{X}}^{(t)}} \right) \right) + {\eta}\left( {{\bf{X}}^{(t)}} \right)$ and ${\boldsymbol{\theta}} ( \cdot )$ denotes the eigenvector corresponding to the largest eigenvalue of a matrix.

Therefore, given points ${\bf{V}}^{(t)}$ and ${\bf{W}}_{n,k}^{(t)}, \forall n,k$, by exploiting inequalities \eqref{penalty_lower_bound}, we obtain the convex upper bound of the penalty term $F$, which is given by ${\overline F}^{(t)} \triangleq \sum_{(n,k)} {\left( {{\text{Tr}}\left( {{{\bf{W}}_{n,k}}} \right) - {\overline \eta}\left( {{{\bf{W}}_{n,k}};{\bf{W}}_{n,k}^{(t)}} \right)} \right)} + {\text{Tr}}\left( {\bf{V}} \right) - {\overline \eta}\left( {{\bf{V}};{{\bf{V}}^{(t)}}} \right)$.
By replacing the non-concave objective function of problem \eqref{Problem_CVX_penalty_function1} with its concave lower bound, problem \eqref{Problem_CVX_penalty_function1} is transformed into a convex problem:
\begin{subequations}
\label{Problem_CVX_penalty}
\begin{align}
\label{Problem_CVX_penalty_obj}
&\mathop {\max}_{{\bf{V}}, \left\{ {{\bf{W}}_{n,k}}, {u_{n,k}} \right\}} \;
\sum_{(n,k)} \left( {{s_{n,k}\left( {\bf{V}}, {\bf{W}} \right)} - {u_{n,k}}} \right) - \xi {\overline F ^{(t)}}\\
&\;\;\;\;\;\;\;\;{\rm{s}}.{\rm{t}}.\;\;\;\;\;\;\;\;\;
\eqref{Problem_CVX_cons_inter_temp} - \eqref{Problem_CVX_cons_semidefinite}, \eqref{Problem_CVX_auxi_vari_v_Taylor}. \nonumber
\end{align}
\end{subequations}
This convex problem can be solved by the existing tools such as CVX \cite{grant2011cvx}.
Since both objective function and constraint \eqref{Problem_CVX_auxi_vari_v_Taylor} involve the Taylor expansion, we can update the points of Taylor expansion after solving problem \eqref{Problem_CVX_penalty} and repeat this process until the objective function converges.



\subsection{Penalty-Based Iterative Beamforming Scheme}
The PIBF scheme for solving problem \eqref{Problem_CVX_SDR} consists of two loops: the outer loop and the inner loop.
In each outer loop, the penalty factor is updated by $\xi=\omega \xi$, where $\omega > 1$. 
The outer loop terminates when the violation of constraint \eqref{cons_rank1_equi}, i.e., the penalty term $F$, is less than a threshold.
In each inner loop, given the penalty factor, we iteratively solve problem \eqref{Problem_CVX_penalty} and update the points of Taylor expansion ${u_{n,k}^{(t)}}, \forall n,k$, ${\bf{V}}^{(t)}$, and ${\bf{W}}_{n,k}^{(t)}, \forall n,k$.
%
%
Specifically, in the $t$-th inner loop, after solving problem \eqref{Problem_CVX_penalty}, we update ${\bf{V}}^{(t+1)}$ and ${\bf{W}}_{n,k}^{(t+1)}, \forall n,k$ with the solution of problem \eqref{Problem_CVX_penalty}.
Then, we update $u_{n,k}^{(t+1)}, \forall n,k$ by making the equality in \eqref{Problem_CVX_auxi_vari_v} holds, i.e.,
\begin{align}
\label{update_v_t}
u_{n,k}^{(t+1)} = \ln \left( {\alpha_{n,k}}\left( {\bf{V}}^{(t+1)}, {\bf{W}}^{(t+1)} \right) \right),
\end{align}
where ${\bf{W}}^{(t)}$ represents the collection of all ${\bf{W}}_{n,k}^{(t)}$.
Updating $u_{n,k}^{(t+1)}, \forall n,k$ by \eqref{update_v_t} ensures convergence of the inner loop, which is described in detail in Section \ref{subsubsec_Convergence}.
The inner loop terminates when the objective function of problem \eqref{Problem_CVX_penalty} converges.
            
After the outer loop terminates, we obtain the beamforming matrices ${\bf{V}}^*$ and ${\bf{W}}_{n,k}^*, \forall n,k$, which are approximately rank-one matrices.
Thus, after performing singular value decomposition (SVD) for ${\bf{V}}^*$ and ${\bf{W}}_{n,k}^*, \forall n,k$, we can recover ${\bf{v}}^*$ and ${\bf{w}}_{n,k}^*, \forall n,k$ by ${\bf{v}}^* = \sqrt {{\eta}\left( {{{\bf{V}}^{*}}} \right)} {\boldsymbol{\theta}} \left( {{{\bf{V}}^{*}}} \right)$ and ${{\bf{w}}_{n,k}^*} \!=\! \sqrt {{\eta}\left( {{\bf{W}}_{n,k}^{*}} \right)} {\boldsymbol{\theta}} \!\left( {{\bf{W}}_{n,k}^{*}} \right)\!, \forall n,k$, respectively.
The PIBF scheme to solve problem \eqref{Problem_sate_active} is summarized in Algorithm \ref{algorithm_penalty_IBF_LYW}.




\begin{algorithm}
\caption{Penalty-Based Iterative Beamforming Scheme}
\label{algorithm_penalty_IBF_LYW}
\begin{algorithmic}[1]
    \STATE {Set $\varepsilon_1>0$, $\varepsilon _2>0$, the maximum number of iterations ${t_{\max}}$, the penalty factor $\xi>0$, and $\omega>1$.}
    \STATE {Initialize the algorithm with feasible points ${{\bf{V}}^{\left( 0 \right)}}$, ${\bf{W}}_{n,k}^{(0)}, \forall n,k$, and $u_{n,k}^{\left( 0 \right)}, \forall n,k$.}
    
    \REPEAT 
    \STATE {Set the iteration number $t=0$ and $\phi^{(0)}=-\infty$.}
    \REPEAT 
    \STATE {Given ${{\bf{V}}^{(t)}}$, ${\bf{W}}_{n,k}^{(t)}, \forall n,k$, and ${u_{n,k}^{(t)}}, \forall n,k$, solve problem \eqref{Problem_CVX_penalty} by CVX and obtain the value of \eqref{Problem_CVX_penalty_obj}, i.e., $\phi^{(t+1)}$ and the solution ${\bf{V}}$, ${\bf{W}}_{n,k}, \forall n,k$, and ${u_{n,k}}, \forall n,k$.}
    \STATE {Set ${{\bf{V}}^{(t+1)}} = {\bf{V}}$, ${\bf{W}}_{n,k}^{(t+1)} = {\bf{W}}_{n,k}, \forall n,k$.}
    \STATE {Update $u_{n,k}^{(t+1)}, \forall n,k$ with ${{\bf{V}}^{(t+1)}}$ and ${\bf{W}}_{n,k}^{(t+1)}, \forall n,k$ by \eqref{update_v_t}.}

    \STATE {$t=t+1$.}
    \UNTIL {$|\phi^{(t)}-\phi^{(t-1)}| \le \varepsilon_1$ or $t = {t_{\max}}$.}
    \STATE Set $\xi=\omega \xi$, ${{\bf{V}}^{(0)}}={{\bf{V}}^{(t)}}$, ${\bf{W}}_{n,k}^{(0)}={\bf{W}}_{n,k}^{(t)}, \forall n,k$, and ${u_{n,k}^{(0)}} = {u_{n,k}^{(t)}}, \forall n,k$.
     \STATE Calculate $F$ with ${{\bf{V}}^{(t)}}$ and ${\bf{W}}_{n,k}^{(t)}, \forall n,k$.
   \UNTIL {The violation of constraint \eqref{cons_rank1_equi} is below the predefined threshold $\varepsilon _2$, i.e., $F  < \varepsilon _2$.}

    \STATE {Set ${{\bf{V}}^{*}}={{\bf{V}}^{(t)}}$ and ${\bf{W}}_{n,k}^{*}={\bf{W}}_{n,k}^{(t)}, \forall n,k$.}
    \STATE Recover ${\bf{v}}^*$ and ${\bf{w}}_{n,k}^*, \forall n,k$ by performing SVD for ${\bf{V}}^*$ and ${\bf{W}}_{n,k}^*, \forall n,k$, respectively.
\end{algorithmic}
\end{algorithm}

\subsection{Initialization Algorithm}
In order to improve the convergence speed of Algorithm \ref{algorithm_penalty_IBF_LYW}, we aim to initialize Algorithm \ref{algorithm_penalty_IBF_LYW} with feasible Taylor expansion points ${{\bf{V}}^{\left( 0 \right)}}$, ${\bf{W}}_{n,k}^{(0)}, \forall n,k$, and $u_{n,k}^{\left( 0 \right)}, \forall n,k$.
Inspired by \cite{lin2019joint}, to obtain the feasible initial points, we introduce a variable $\delta$ to measure the satisfaction of all constraints.
For example, we relax \eqref{Problem_CVX_cons_aerial_BS}, i.e., ${\overline q} - {\rm{Tr}}\left( {\bf{V}} \right) \ge 0$, to ${\overline q} - {\rm{Tr}}\left( {\bf{V}} \right) \ge \delta$, which is the constraint \eqref{Problem_CVX_cons_aerial_BS}$^*$.
When $\delta < 0$, constraint \eqref{Problem_CVX_cons_aerial_BS} may not be satisfied. In this case, if $\delta$ is larger, the largest possible violation of constraint \eqref{Problem_CVX_cons_aerial_BS} is more minor.
When $\delta \ge 0$, constraint \eqref{Problem_CVX_cons_aerial_BS} is satisfied.
Thus, by maximizing $\delta$, we can maximize the satisfaction of constraint \eqref{Problem_CVX_cons_aerial_BS}.
Using the same way, we obtain \eqref{Problem_CVX_cons_inter_temp}$^*$ - \eqref{Problem_CVX_cons_semidefinite}$^*$ and \eqref{Problem_CVX_auxi_vari_v_Taylor}$^*$. 
%
To find a feasible solution to problem \eqref{Problem_CVX_penalty}, we aim to maximize the satisfaction of all constraints by formulating the following problem:
\begin{subequations}
\label{Problem_CVX_initialization}
\begin{align}
    &\mathop {\max}_{ {{\bf{V}}, \left\{ {{\bf{W}}_{n,k}}, {u_{n,k}} \right\}}, \delta} \; \delta \\
    &\;\;\;\;\;\;\;\;\;{\rm{s}}.{\rm{t}}.\;\;\;\;\;\;\;\;\;\;
    \eqref{Problem_CVX_cons_inter_temp}^* - \eqref{Problem_CVX_cons_semidefinite}^* , \eqref{Problem_CVX_auxi_vari_v_Taylor}^*. \nonumber
\end{align}
\end{subequations}


Since \eqref{Problem_CVX_cons_inter_temp}$^*$ - \eqref{Problem_CVX_cons_semidefinite}$^*$ and \eqref{Problem_CVX_auxi_vari_v_Taylor}$^*$ are all convex, this convex problem can be solved by CVX \cite{grant2011cvx}.
By iteratively solving the initialization problem \eqref{Problem_CVX_initialization}, we can obtain feasible initial points ${{\bf{V}}^{\left( 0 \right)}}$, ${\bf{W}}_{n,k}^{(0)}, \forall n,k$, and $u_{n,k}^{\left( 0 \right)}, \forall n,k$ in Algorithm \ref{algorithm_penalty_IBF_LYW}.
The initialization algorithm is summarized in Algorithm \ref{algorithm_Initial}.

\begin{algorithm}
\caption{Initialization Algorithm}
\label{algorithm_Initial}
\begin{algorithmic}[1]
    \STATE {Set the iteration number $t=0$.}
    \STATE {Initialize the algorithm with random points $u_{n,k}^{\left( 0 \right)}, \forall n,k$.}
    \REPEAT 
    \STATE {Given ${u_{n,k}^{(t)}}, \forall n,k$, solve problem \eqref{Problem_CVX_initialization} and obtain the solution ${\bf{V}}$, ${{\bf{W}}_{n,k}}, \forall n,k$, and ${u_{n,k}}, \forall n,k$ by CVX.}
    \STATE {Set $u_{n,k}^{(t+1)} = \ln \left( {\alpha_{n,k}}\left( {\bf{V}}, {\bf{W}} \right) \right), \forall n,k$.}
    \STATE {$t=t+1$.}
    \UNTIL {$\delta \ge 0$.}
    \STATE Set ${\bf{V}}^{(0)} = {\bf{V}}$, ${\bf{W}}_{n,k}^{(0)} = {\bf{W}}_{n,k}, \forall n,k$, and ${u_{n,k}^{(0)}} = {u_{n,k}^{(t)}}, \forall n,k$.
    \STATE Output feasible points ${{\bf{V}}^{(0)}}$, ${\bf{W}}_{n,k}^{(0)}, \forall n,k$, and ${u_{n,k}^{(0)}}, \forall n,k$.
\end{algorithmic}
\end{algorithm}



\subsection{Convergence and Complexity Analysis}
In this section, we will analyze the convergence performance and computational complexity of the PIBF scheme.
\subsubsection{Convergence Analysis}
\label{subsubsec_Convergence}
The purpose of the outer loop in Algorithm \ref{algorithm_penalty_IBF_LYW} is to make the rank-one constraint \eqref{cons_rank1_equi} be satisfied enough.
Thus, we will present the effectiveness of Algorithm \ref{algorithm_penalty_IBF_LYW} in simulation.
Here, we will focus on the convergence performance for the inner loop in Algorithm \ref{algorithm_penalty_IBF_LYW}.

To facilitate the analysis of convergence for the inner loop in the PIBF scheme, we define an objective function consisting of a sum rate term minus a penalty term, which is given by: 
\begin{align}
\label{conver_obj}
\mu\!\left( {\bf{V}}, \!{\bf{W}} \right) \!\triangleq\!\! \sum_{(n,k)} \!\!\left( s_{n,k}\left( {\bf{V}}, \!{\bf{W}} \right)\! -\! \ln \left({\alpha_{n,k}}\!\left( {\bf{V}}, \!{\bf{W}} \right)\right) \right)\! -\! \xi F\!\left( {\bf{V}},\! {\bf{W}} \right)\!.
\end{align}

In addition, the solution of problem \eqref{Problem_CVX_penalty} obtained by CVX in the $t$-th iteration is denoted by $\widetilde {\bf{W}}_{n,k}^{(t)}, \forall n,k$, $\widetilde {\bf{V}}^{(t)}$, and $\widetilde u_{n,k}^{(t)}, \forall n,k$.
When this solution of problem \eqref{Problem_CVX_penalty} is obtained in the $t$-th iteration, the objective function of problem \eqref{Problem_CVX_penalty} is denoted by:
\begin{align}
\label{conver_obj_penalty}
&\phi\left( {{\widetilde {\bf{V}}}^{(t)}}, {{\widetilde {\bf{W}}}^{(t)}}, {\widetilde u^{(t)}}; {{\bf{V}}^{(t)}}, {{\bf{W}}^{(t)}}, {u^{(t)}} \right) \nonumber \\
&\;\;\;\;\;\;\;\;\;\;\;\;\triangleq 
\sum_{(n,k)} {\left( {s_{n,k}\left( \widetilde {\bf{V}}^{(t)}, \widetilde {\bf{W}}^{(t)} \right) - \widetilde u_{n,k}^{(t)}} \right)}  \nonumber \\
&\;\;\;\;\;\;\;\;\;\;\;\;\;\;\;\;- \xi \overline F\left( {{{\widetilde {\bf{V}}}^{(t)}}, {{\widetilde {\bf{W}}}^{(t)}};{{\bf{V}}^{(t)}}, {{\bf{W}}^{(t)}}} \right),
\end{align}
where $\widetilde {\bf{W}}^{(t)}$,  ${\widetilde u^{(t)}}$, and $u^{(t)}$ represent the collection of all $\widetilde {\bf{W}}_{n,k}^{(t)}$, all $\widetilde u_{n,k}^{(t)}$, and all $u_{n,k}^{(t)}$, respectively.


To prove the convergence for the inner loop in Algorithm \ref{algorithm_penalty_IBF_LYW}, we present the following theorem \cite{sun2016majorization}:
\begin{thm}
\label{theorem_convergence}
As the number of the inner loop iterations increases, the solution of problem \eqref{Problem_CVX_penalty} makes $\mu\left(  {\bf{V}}, {\bf{W}} \right)$ to be a non-decreasing function, i.e., $\mu\left( {{\bf{V}}^{(t + 1)}}, {{\bf{W}}^{(t + 1)}} \right) \ge \phi({\widetilde {\bf{V}}^{(t)}}, {\widetilde {\bf{W}}^{(t)}}, {{\tilde u}^{(t)}};{{\bf{V}}^{(t)}}, {{\bf{W}}^{(t)}}, {u^{(t)}}) \ge \mu\left( {{{\bf{V}}^{(t)}}, {{\bf{W}}^{(t)}}} \right)$.
\end{thm} 

\begin{proof}
Please refer to the Appendix A.
\end{proof}

Theorem \ref{theorem_convergence} indicates that as the number of the inner loop iterations increases, the objective function of problem \eqref{Problem_CVX_penalty}, i.e., \eqref{conver_obj_penalty} increases.
Thus, the inner loop of Algorithm \ref{algorithm_penalty_IBF_LYW} is convergent.

\subsubsection{Complexity Analysis}
\label{Complexity_Analysis_proposed}
The complexity of Algorithm \ref{algorithm_penalty_IBF_LYW} is mainly affected by the complexity of solving problem \eqref{Problem_CVX_penalty}, which is analyzed as follows.
Since $s_{n,k}\left( {\bf{V}}, {\bf{W}} \right)$ involves a logarithmic function in the objective function \eqref{Problem_CVX_penalty_obj}, problem \eqref{Problem_CVX_penalty} is a generalized nonlinear convex problem, whose complexity is difficult to analyze directly \cite{li2020cooperative}.
In fact, CVX handles the logarithmic function by a successive approximation heuristic \cite{grant2011cvx}. 
By approximating $s_{n,k}\left( {\bf{V}}, {\bf{W}} \right)$ with a sequence of second order cone (SOC), problem \eqref{Problem_CVX_penalty} can be transformed into a SOC programming (SOCP) problem.
This problem can be solved by the interior-point method with complexity ${\cal{O}}\!\left( {(\widetilde K {M_G} +\! {M_A} \!+\! \widetilde K)}^{3.5} \log \left( 1/\epsilon  \right) \right)$ given a solution accuracy $\epsilon > 0$ \cite{li2020cooperative}, where $\widetilde K = \sum_{n = 1}^N {{K_n}}$ denotes the total number of terrestrial terminals.
We can find that when the number of antennas of the terrestrial BSs and the aerial BS, i.e., $M_G$ and $M_A$, increases, the complexity of solving problem \eqref{Problem_CVX_penalty} will be high.


\section{Low Complexity Beamforming Schemes}
\label{sec_Low_Com_Scheme}
Considering the high complexity of the PIBF scheme, in this section, we propose three low-complexity schemes, namely, interference suppression (IS) scheme, zero-forcing (ZF) scheme, and maximal ratio combining (MRC) scheme, to reduce the complexity of solving problem \eqref{Problem_sate_active}.

\subsection{Basic Principle}
\label{sec_Low_Com_Scheme_Principle}
We have seen that the complexity of the PIBF scheme is high, which is mainly due to the high complexity of solving problem \eqref{Problem_CVX_penalty} with a large number of variables.
To reduce the number of variables of the optimization problem, we decouple the beamforming vectors into two parts:
\begin{subequations}
\label{beamforming_vector_rewrite}
\begin{align}
&{\bf{v}} = \sqrt {{q}} {\overline {\bf{v}}},\\
&{{\bf{w}}_{n,k}} = \sqrt {p_{n,k}} {\overline {\bf{w}} _{n,k}}, \forall n,k,
\end{align}
\end{subequations}
where $\overline {\bf{v}}  = {\bf{v}}/{\left\| {\bf{v}} \right\|_2}$ and ${\overline {\bf{w}} _{n,k}} = {{\bf{w}}_{n,k}}/{\left\| {{{\bf{w}}_{n,k}}} \right\|_2}, \forall n,k$ are the normalized part of ${\bf{v}}$ and ${{\bf{w}}_{n,k}}, \forall n,k$, respectively.
Besides, $q = \left\| {\bf{v}} \right\|_2^2$ and ${p_{n,k}} = \left\| {{{\bf{w}}_{n,k}}} \right\|_2^2, \forall n,k$ are the power part of ${\bf{v}}$ and ${{\bf{w}}_{n,k}}, \forall n,k$, respectively.
To design ${\overline {\bf{v}}}$, ${\overline {\bf{w}} _{n,k}}, \forall n,k$, ${q}$, and ${p_{n,k}}, \forall n,k$, we design a two-step algorithm, where ${\overline {\bf{v}}}$ and ${\overline {\bf{w}} _{n,k}}, \forall n,k$ are designed in step 1, then ${q}$ and ${p_{n,k}}, \forall n,k$ are designed in step 2.
Finally, ${\bf{v}}$ and ${{\bf{w}}_{n,k}}, \forall n,k$ can be obtained by \eqref{beamforming_vector_rewrite}.
Note that CVX is only used in step 2 to optimize ${q}$ and ${p_{n,k}}, \forall n,k$ but not in step 1.
Thus, the three proposed schemes have lower complexity compared with the PIBF scheme.



\subsection{Interference Suppression Scheme}
\label{sec_Low_Complexity_Beamforming_Scheme}
Based on the above principle, we use the two-step algorithm to design ${\overline {\bf{v}}}$, ${\overline {\bf{w}} _{n,k}}, \forall n,k$, ${q}$, and ${p_{n,k}}, \forall n,k$.

\subsubsection{Step 1}
In the IS scheme, each normalized beamforming vector is designed individually.
The goal of this design is not only to suppress the sum of the interference (excluding the interference to the satellite terminal) within a given threshold but also to maximize the signal strength over the channel to the intended receiver.

Specifically, the sum of the matrices of the interference channels is given by ${\bf{D}} = {{\bf{H}}_{A,n}} + \sum_{(m,l) \ne (n,k)} {{{\bf{H}}_{n,m,l}}}$, which includes the interference from the terrestrial BS in the $n$-th cell to all other receivers in the terrestrial network and the aerial user.
Inspired by \cite{lin2020robust}, to design ${\overline {\bf{w}} _{n,k}}$, we formulate an optimization problem as follows: 
\begin{subequations}
\label{Problem_Interference_Suppression_step1_1}
\begin{align}
\label{Problem_Interference_Suppression_step1_1_1}
&{\max _{\left\{ {{\overline {\bf{w}} }_{n,k}} \right\}}}\;
\overline {\bf{w}} _{n,k}^H {{\bf{H}}_{n,n,k}} {\overline {\bf{w}} _{n,k}}\\
&\;\;{\rm{s}}.{\rm{t}}.\;\;\;
\label{Problem_Interference_Suppression_step1_1_2}
{\left\| {{\overline {\bf{w}} }_{n,k}} \right\|_2} = 1, \\
\label{Problem_Interference_Suppression_step1_1_3}
&\;\;\;\;\;\;\;\;\;\;\overline {\bf{w}} _{n,k}^H{\bf{D}}{\overline {\bf{w}} _{n,k}} \le \chi,
\end{align}
\end{subequations}
where $\chi$ is an interference threshold that is small enough.
%
In a similar way, we can obtain ${\overline {\bf{v}}}$ by solving the following problem: 
\begin{subequations}
\label{Problem_Interference_Suppression_step1_0}
\begin{align}
&{\max _{ {\overline {\bf{v}}} }}\;
{\overline {\bf{v}}}^H{{\bf{G}}_{A}}{\overline {\bf{v}}}  \\ 
&{\rm{s}}.{\rm{t}}.\;\;\;
{\left\| {\overline {\bf{v}}} \right\|_2} = 1,  \\ 
&\;\;\;\;\;\;\;\;{\overline {\bf{v}}}^H{\bf{D}}{\overline {\bf{v}}} \leq \chi, 
\end{align}
\end{subequations}
where ${\bf{D}} = \sum_{(n,k)} {{\bf{G}}_{n,k}}$ denotes the interference from the aerial BS to all receivers in the terrestrial network. 
Due to the similarity between problem \eqref{Problem_Interference_Suppression_step1_0} and problem \eqref{Problem_Interference_Suppression_step1_1}, problem \eqref{Problem_Interference_Suppression_step1_0} is solved in a similar way to problem \eqref{Problem_Interference_Suppression_step1_1}, we only focus on solving problem \eqref{Problem_Interference_Suppression_step1_1} for brevity.

In order to solve problem \eqref{Problem_Interference_Suppression_step1_1}, we aim to transform it into a generalized Rayleigh quotient form.
Firstly, the inequality constraint \eqref{Problem_Interference_Suppression_step1_1_3} can be transformed into an equality constraint by introducing a variable $\rho$ satisfying $0 \!\leq \!\rho \!\leq \!1$.
Then, problem \eqref{Problem_Interference_Suppression_step1_1} can be transformed into the following equivalent problem:
\begin{subequations}
\label{Problem_Interference_Suppression_step1_2}
\begin{align}
\label{Problem_Interference_Suppression_step1_2_1}
&{\max _{\left\{ {{\overline {\bf{w}} }_{n,k}}\right\}, \rho  }}\;
\overline {\bf{w}} _{n,k}^H{{\bf{H}}_{n,n,k}}{\overline {\bf{w}} _{n,k}}\\
&\;\;\;{\rm{s}}.{\rm{t}}.\;\;\;\;
\label{Problem_Interference_Suppression_step1_2_2}
{\left\| {{\overline {\bf{w}} }_{n,k}} \right\|_2} = 1, \\
\label{Problem_Interference_Suppression_step1_2_3}
&\;\;\;\;\;\;\;\;\;\;\;\;{\chi ^{ - 1}}\overline {\bf{w}} _{n,k}^H{\bf{D}}{\overline {\bf{w}} _{n,k}} + \rho  = 1, \\
&\;\;\;\;\;\;\;\;\;\;\;\;0 \leq \rho \leq 1.
\end{align}
\end{subequations}

By iteratively obtaining $\rho$ and ${\overline {\bf{w}} _{n,k}}$, we can solve \eqref{Problem_Interference_Suppression_step1_2}.
When $\rho$ is given, by substituting \eqref{Problem_Interference_Suppression_step1_2_2} and \eqref{Problem_Interference_Suppression_step1_2_3} into \eqref{Problem_Interference_Suppression_step1_2_1}, \eqref{Problem_Interference_Suppression_step1_2} is transformed into a generalized Rayleigh quotient form:
\begin{subequations}
\label{Problem_Interference_Suppression_step1_3}
\begin{align}
&{\max _{\left\{ {{{\overline {\bf{w}} }_{n,k}} } \right\}}}\;
\frac{{\overline {\bf{w}} _{n,k}^H{{\bf{H}}_{n,n,k}}{{\overline {\bf{w}} }_{n,k}}}}{{\overline {\bf{w}} _{n,k}^H\left( {{\bf{D}}{\chi ^{ - 1}} + \rho {\bf{I}}} \right){{\overline {\bf{w}} }_{n,k}}}}\\
&\;\;{\rm{s}}.{\rm{t}}.\;\;\;
{\left\| {{\overline {\bf{w}} }_{n,k}} \right\|_2} = 1.
\end{align}
\end{subequations}

Then, the optimal solution of problem \eqref{Problem_Interference_Suppression_step1_3} and the corresponding objective function value $\psi$ are respectively given by:
\begin{align}
\label{Calculate_max_eigenvector}
{\overline {\bf{w}} _{n,k}}\left( {{\rho}} \right) 
&= {\boldsymbol{\theta}} \left( {{{\bf{H}}_{n,n,k}}, {\bf{D}}\chi^{-1} + {\rho}{\bf{I}}} \right),\\
\label{Calculate_max_eigenvalue}
{\psi}\left( {{\rho}} \right) 
&= {\eta}\left( {{{\bf{H}}_{n,n,k}}, {\bf{D}}\chi^{-1} + {\rho}{\bf{I}}} \right),
\end{align}
where ${\eta}\left( {\bf{A}}, {\bf{B}} \right)$ and ${\boldsymbol{\theta}} \left( {\bf{A}}, {\bf{B}} \right)$ denote the largest generalized eigenvalue and the corresponding eigenvector of $\left( {\bf{A}}, {\bf{B}} \right)$.
Note that since ${{\bf{H}}_{n,n,k}}$ and ${\bf{D}}\chi^{-1} + {\rho}{\bf{I}}$ are both positive definite hermite matrices, $\psi$ is real.

When ${\overline {\bf{w}} _{n,k}}$ is given, $\rho$ is updated by the equality constraint \eqref{Problem_Interference_Suppression_step1_2_3}, i.e., $\rho = 1 - {\overline {\bf{w}} _{n,k} ^H}{\bf{D}}\overline {\bf{w}} _{n,k}\chi^{-1}$.
By iteratively solving $\rho$ and ${\overline {\bf{w}} _{n,k}}$ until $\psi$ converges, problem \eqref{Problem_Interference_Suppression_step1_2} is solved.
The algorithm for solving problem \eqref{Problem_Interference_Suppression_step1_2} is summarized in Algorithm \ref{algorithm_interference_suppression}. 
To obtain ${\overline {\bf{v}}}$ and ${\overline {\bf{w}} _{n,k}}, \forall n,k$ in step 1 of the IS scheme, we solve the corresponding problem \eqref{Problem_Interference_Suppression_step1_1} with Algorithm \ref{algorithm_interference_suppression} for each ${\overline {\bf{w}} _{n,k}}$ and solve problem \eqref{Problem_Interference_Suppression_step1_0} in a similar way for ${\overline {\bf{v}}}$.




\begin{algorithm}
\caption{Algorithm for solving problem \eqref{Problem_Interference_Suppression_step1_2}}
\label{algorithm_interference_suppression}
\begin{algorithmic}[1]
    \STATE {Set $\varepsilon_3>0$, the iteration number $t=0$, the threshold $\chi$, and $\rho^{\left( 0 \right)} = 0$.}
    \STATE {Calculate ${\overline{\bf{w}}}_{n,k}^{\left( 0 \right)}$ and $\psi^{\left( 0 \right)}$ with $\rho^{\left( 0 \right)}$ by \eqref{Calculate_max_eigenvector} and \eqref{Calculate_max_eigenvalue}, respectively.}
    \REPEAT 
    \STATE {Calculate $\rho^{(t+1)}$ by $\rho^{(t+1)} = 1 - {\left( {\overline {\bf{w}} _{n,k}^{(t)}} \right)^H}{\bf{D}}\overline {\bf{w}} _{n,k}^{(t)}\chi^{-1}$.}
    \STATE {Calculate ${\overline{\bf{w}}}_{n,k}^{(t+1)}$ and $\psi^{(t+1)}$ with $\rho^{(t+1)}$ by \eqref{Calculate_max_eigenvector} and \eqref{Calculate_max_eigenvalue}, respectively.}
    \STATE {$t=t+1$.}
    \UNTIL {${\left| {\psi^{(t)} - \psi^{\left( t - 1 \right)}} \right|}  \leq \varepsilon_3$.}
    \STATE {Output ${\overline {\bf{w}} _{n,k}} = \overline {\bf{w}} _{n,k}^{(t)}$.}
\end{algorithmic}
\end{algorithm}

\subsubsection{Step 2}
When ${\overline {\bf{v}}}$ and ${\overline {\bf{w}} _{n,k}}, \forall n,k$ are given, we substitute them into problem \eqref{Problem_sate_active} and obtain an optimization problem with respect to ${q}$ and ${p_{n,k}}, \forall n,k$.
For the convenience of presentation, we denote ${\overline {\bf{V}}} = {\overline {\bf{v}}}{\overline {\bf{v}}}^H$ and ${{\overline{\bf W}}_{n,k}} = {\overline {\bf{w}} _{n,k}}\overline {\bf{w}} _{n,k}^H, \forall n,k$.
From \eqref{beamforming_vector_rewrite}, we have ${\bf{V}} = q{\overline {\bf{V}}}$ and ${{\bf{W}}_{n,k}} = p_{n,k}{\overline {\bf{W}} _{n,k}}$.
Besides, we denote the collection of all $p_{m,l}$ by $p$ and define ${\overline s _{n,k}}\left( {q,p} \right) \triangleq \ln \Big({\rm{Tr}}\left( {{{\bf{G}}_{n,k}}\overline {\bf{V}} } \right)q + \sum_{(m,l)} {{\rm{Tr}}\left( {{{\bf{H}}_{m,n,k}}{{{\bf{\overline W}}}_{m,l}}} \right){p_{m,l}}}  + \overline \sigma _{n,k}^2\Big)$, which is a concave function.
Then, problem \eqref{Problem_sate_active} can be rewritten as:
\begin{subequations}
\label{Problem_find_power_considering_interference1}
\begin{align}
\label{sum_rate_wp}
&\mathop {\max}_{q, \left\{ {p_{n,k}}\right\}} \;
\sum_{(n,k)} {\overline s _{n,k}}\left( {q,p} \right) - \sum_{(n,k)} \ln \left( {{\overline \alpha}_{n,k}}\left( q,p \right) \right) \\
&\;\;\;{\rm{s}}.{\rm{t}}.\;
\label{cons_interference_power_sate_wp}
\sum_{(n,k)} \!{{\rm{Tr}}\!\left( {{{\bf{H}}_{S,n}}{{{\bf{\overline W}}}_{n,k}}} \right){p_{n,k}}}+  {\rm{Tr}}\left( {{{\bf{G}}_{S}}{\overline {\bf{V}}}} \right){q} \leq {{\overline I}_S},  \\ 
\label{cons_SINR_wp}
&\;\;\;\;\;\;\;\;\;{\rm{Tr}} \! \left( {{\bf{G}}_{A}} {\overline {\bf{V}}} \right) \! {q}  \! \geq  \!  {\overline \beta}  \! \left(  \! {\overline \sigma }^2_A  \! + \!  \sum_{(m,l)}  \! {{\rm{Tr}} \! \left( {{{\bf{H}}_{A,m}}{{{\bf{\overline W}}}_{m,l}}} \right) \! {p_{m,l}}}   \! \right) \! ,\\
\label{cons_transmit_power_terrestrial_BS_wp}
&\;\;\;\;\;\;\;\;\;\sum_{k = 1}^{{K_n}} {{p_{n,k}}}  \leq {{\overline p}_n},\forall n, \\ 
\label{cons_transmit_power_aerial_BS_wp}
&\;\;\;\;\;\;\;\;\;{q} \leq {\overline q},  \\ 
\label{cons_transmit_power_0}
&\;\;\;\;\;\;\;\;\;{p_{n,k}} \geq 0, \forall n,k, {q} \geq 0,  
\end{align}
\end{subequations}
where ${{\overline \alpha}_{n,k}}\left( q,p \right) \triangleq {\rm{Tr}}\left( {{{\bf{G}}_{n,k}}\overline {\bf{V}} } \right)q + \sum_{(m,l) \ne (n,k)} {{\rm{Tr}}\left( {{{\bf{H}}_{m,n,k}}{{{\bf{\overline W}}}_{m,l}}} \right){p_{m,l}}}  + \overline \sigma _{n,k}^2$.

Since $\sum_{(n,k)} \ln \left( {{\overline \alpha}_{n,k}}\left( q,p \right) \right)$ in the objective function \eqref{sum_rate_wp} is concave, \eqref{sum_rate_wp} is non-concave.
To deal with this, we aim to maximize \eqref{sum_rate_wp} by maximizing its concave lower bound.
Similar to the transformation of problem \eqref{Problem_CVX_SDR} into problem \eqref{Problem_CVX_SDR_auxi_vari}, we introduce variables ${u_{n,k}}, \forall n,k$ and transform problem \eqref{Problem_find_power_considering_interference1} into:
\begin{subequations}
\label{Problem_find_power_considering_interference2}
\begin{align}
&\mathop {\max}_{{q}, \left\{ {p_{n,k}}, {u_{n,k}} \right\}} \;
\sum_{(n,k)} \left( {\overline s _{n,k}}\left( {q,p} \right) - {u_{n,k}} \right)    \\ 
&\;\;\;\;\;\;{\rm{s}}.{\rm{t}}.\;\;\;\;\;\;\;\;
\eqref{cons_interference_power_sate_wp} - \eqref{cons_transmit_power_0}, \nonumber \\ 
\label{Problem_find_power_v_noncon}
&\;\;\;\;\;\;\;\;\;\;\;\;\;\;\;\;\;\;\;{{\overline \alpha}_{n,k}}\left( q,p \right) \le {e^{u_{n,k}}}, \forall n,k.
\end{align}
\end{subequations}

Nonetheless, the constraint \eqref{Problem_find_power_v_noncon} is still non-convex due to the exponential function on the right-hand side of the inequality.
Thus, similar to the transformation of the constraint \eqref{Problem_CVX_auxi_vari_v} to the convex constraint \eqref{Problem_CVX_auxi_vari_v_Taylor}, we apply the first-order Taylor expansion to ${e^{u_{n,k}}}, \forall n,k$ around ${u_{n,k}^{(t)}}, \forall n,k$ in the $t$ iteration.
Then, compared with constraint \eqref{Problem_find_power_v_noncon}, we obtain a tighter constraint:
\begin{align}
\label{Problem_find_power_v}
{{\overline \alpha}_{n,k}}\left( q,p \right) \le {e^{u_{n,k}^{(t)}}}\left( {{u_{n,k}} - u_{n,k}^{(t)} + 1} \right), \forall n,k.
\end{align}
%
By replacing constraint \eqref{Problem_find_power_v_noncon} with \eqref{Problem_find_power_v}, problem \eqref{Problem_find_power_considering_interference2} is transformed into:
\begin{subequations}
\label{Problem_find_power_considering_interference}
\begin{align}
\label{Problem_find_power_obj}
&\mathop {\max}_{{q}, \left\{ {p_{n,k}}, {u_{n,k}} \right\}} \;
\sum_{(n,k)} \left( {\overline s _{n,k}}\left( {q,p} \right) - {u_{n,k}} \right)    \\ 
&\;\;\;\;\;\;{\rm{s}}.{\rm{t}}.\;\;\;\;\;\;\;\;
\eqref{cons_interference_power_sate_wp} - \eqref{cons_transmit_power_0}, \eqref{Problem_find_power_v}. \nonumber 
\end{align}
\end{subequations}
%
Problem \eqref{Problem_find_power_considering_interference} is convex and can be easily solved by CVX \cite{grant2011cvx}.
Note that when the optimal solution of problem \eqref{Problem_find_power_considering_interference} is obtained, the equalities in \eqref{Problem_find_power_v} hold.
This is because if the left side of \eqref{Problem_find_power_v} is less than the right side, then $u_{n,k}$ will decrease until the equalities in \eqref{Problem_find_power_v} hold.
Since the constraint \eqref{Problem_find_power_v} involves the Taylor expansion, we update the points of Taylor expansion $u_{n,k}^{(t)}, \forall n,k$ after solving problem \eqref{Problem_find_power_considering_interference}, and repeat this process until the objective function of problem \eqref{Problem_find_power_considering_interference} converges.
The algorithm to find ${q}$ and ${p_{n,k}}, \forall n,k$ in step 2 of the IS scheme is summarized as Algorithm \ref{algorithm_power}.
The initialization algorithm used to obtain the initial points $u_{n,k}^{\left( 0 \right)}, \forall n,k$ is similar to Algorithm \ref{algorithm_Initial}, so it is omitted for brevity.
\begin{algorithm}
\caption{Step 2 of the interference suppression scheme}
\label{algorithm_power}
\begin{algorithmic}[1]
    \STATE {Set $\varepsilon_4>0$, the iteration number $t=0$, and $\overline \phi^{(0)}=-\infty$.}
    \STATE {Initialize the algorithm with feasible points $u_{n,k}^{\left( 0 \right)}, \forall n,k$.}
    \REPEAT 
    \STATE {Given $u_{n,k}^{(t)}, \forall n,k$, solve problem \eqref{Problem_find_power_considering_interference}  by CVX and obtain the value of \eqref{Problem_find_power_obj}, i.e., $\overline \phi^{(t+1)}$ and the solution $q$, $p_{n,k}, \forall n,k$, and ${u_{n,k}}, \forall n,k$.}
    \STATE {Update $u_{n,k}^{(t+1)}, \forall n,k$ by \eqref{update_v_t}, where ${\bf{V}}^{(t + 1)} = q\overline {\bf{V}}$ and ${\bf{W}}_{m,l}^{(t + 1)} = p_{m,l} {\overline {\bf{W}} _{m,l}}, \forall m,l$.}
    \STATE {$t=t+1$.}
    \UNTIL {$|\overline \phi^{(t)}-\overline \phi^{(t-1)}| \le \varepsilon_4$.}
    \STATE Output $q$ and $p_{n,k}, \forall n,k$.
\end{algorithmic}
\end{algorithm}



\subsubsection{Complexity Analysis}
The complexity of step 1 of the IS scheme is ${\cal{O}}( \widetilde K {I_3}( M_G^2 + M_G  + 2M_G^3 ) + {I_3}( M_A^2 + M_A  + 2M_A^3 ))$, where $I_{3}$ is the number of iterations.
This complexity is relatively low since the variables are updated with closed-form expressions and there is no optimization problem involved in Algorithm \ref{algorithm_interference_suppression}.
The complexity of the IS scheme is mainly affected by step 2, i.e., Algorithm \ref{algorithm_power}, whose complexity is mainly affected by solving problem \eqref{Problem_find_power_considering_interference}.
Similar to the analysis in Section \ref{Complexity_Analysis_proposed}, problem \eqref{Problem_find_power_considering_interference} can be transformed into a SOCP problem, which is solved by the interior-point method with complexity ${\cal{O}}\left( {(2\widetilde K  + 1)}^{3.5} \log \left( 1/\epsilon  \right) \right)$ given a solution accuracy $\epsilon > 0$ \cite{li2020cooperative}.
Since problem \eqref{Problem_find_power_considering_interference} has fewer variables than problem \eqref{Problem_CVX_penalty}, the complexity of solving problem \eqref{Problem_find_power_considering_interference} is lower than that of solving problem \eqref{Problem_CVX_penalty}.
Therefore, the complexity of the IS scheme is lower than that of the PIBF scheme.
Nonetheless, since we need to solve problem \eqref{Problem_find_power_considering_interference} and update $u_{n,k}^{(t)}, \forall n,k$ iteratively in Algorithm \ref{algorithm_power}, the complexity is still high.
In the following, we propose the ZF scheme to further reduce the complexity.

\subsection{Zero-Forcing Scheme}
\label{sec_Zero_2orcing_Scheme}
Based on the principle described in section \ref{sec_Low_Com_Scheme_Principle}, we use the two-step algorithm to design ${\overline {\bf{v}}}$, ${\overline {\bf{w}} _{n,k}}, \forall n,k$, ${q}$, and ${p_{n,k}}, \forall n,k$.


\subsubsection{Step 1}
In the ZF scheme, each normalized beamforming variable is designed individually.
The goal of the design is not only to eliminate the interference (excluding the interference to the satellite terminal) but also to maximize the signal strength over the channel to the intended receiver.

Specifically, the interference channels include the interference from the terrestrial BS in the $n$-th cell to all other receivers in the terrestrial network and the aerial user.
Thus, we combine the interference channels into a new matrix ${{\bf{T}}_{n,k}} \in {\mathbb{C}^{{M_G} \times \widetilde K}}$, whose columns are composed of ${{\bf{h}}_{n,m,l}}, \forall (m,l) \ne (n,k)$ and ${{\bf{h}}_{A,n}}$.
By assuming ${M_G} > \widetilde K$, the normalized beamforming vector of the terrestrial BS ${\overline {\bf{w}} _{n,k}}$ can be designed, namely:
\begin{align}
\label{ZF_wnk}
{\bf{T}}_{n,k}^H{\overline {\bf{w}} _{n,k}} = {{\bf{0}}_{{\widetilde K} \times 1}}.
\end{align}

Therefore, ${\overline {\bf{w}} _{n,k}}$ can be designed to maximize ${\left| {{\bf{h}}_{n,n,k}^H{{\overline {\bf{w}} }_{n,k}}} \right|^2}$ and exist in the null-space of ${{\bf{T}}_{n,k}}$: 
\begin{align}
\label{Calculate_ZF_wnk}
{\overline {\bf{w}} _{n,k}} = \frac{{\left( {{{\bf{I}}_{M_G}} - {\bf{T}}_{n,k}^ \bot } \right){{\bf{h}}_{n,n,k}}}}{{{{\left\| {\left( {{{\bf{I}}_{M_G}} - {\bf{T}}_{n,k}^ \bot } \right){{\bf{h}}_{n,n,k}}} \right\|}_2}}},
\end{align}
where ${\bf{T}}_{n,k}^ \bot  = {{\bf{T}}_{n,k}}{\left( {{\bf{T}}_{n,k}^H{{\bf{T}}_{n,k}}} \right)^{ - 1}}{\bf{T}}_{n,k}^H$ is the orthogonal projection matrix of ${{\bf{T}}_{n,k}}$.

Similarly, by constructing a new matrix ${\bf{Q}} \in {{\mathbb C}^{{M_A} \times \widetilde K}}$, whose columns are composed of the interference channels from the aerial BS to all terrestrial terminals, i.e., ${{\bf{g}}_{n,k}}, \forall n,k$, and assuming ${M_A} > \widetilde K$, ${\overline {\bf{v}}}$ can be designed as: 
\begin{align}
\label{Calculate_ZF_wA}    
{\overline {\bf{v}}} = \frac{{\left( {{{\bf{I}}_{M_A}} - {\bf{Q}}^ \bot } \right){{\bf{g}}_{A}}}}{{{{\left\| {\left( {{{\bf{I}}_{M_A}} - {\bf{Q}}^ \bot } \right){{\bf{g}}_{A}}} \right\|}_2}}},
\end{align}
where ${\bf{Q}}^ \bot  = {\bf{Q}}{\left( {{\bf{Q}}^H{\bf{Q}}} \right)^{ - 1}}{\bf{Q}}^H$.
Note that in the ZF scheme, the normalized beamforming vector ${\overline {\bf{w}} _{n,k}}$ multiplied by each interference channel is required to be zero in \eqref{ZF_wnk}.
In contrast, in the IS scheme, ${\overline {\bf{w}} _{n,k}}$ multiplied by the sum of the interference channel matrices ${\bf{D}}$ is required to be less than the threshold in the constraint \eqref{Problem_Interference_Suppression_step1_1_3}.
Since some channels cancel each other out when calculating ${\bf{D}}$ in the IS scheme, the value range of ${\overline {\bf{w}} _{n,k}}$ in the IS scheme is larger than that in the ZF scheme.

\subsubsection{Step 2}


When ${\overline {\bf{v}}}$ and ${\overline {\bf{w}} _{n,k}}, \forall n,k$ are given, ${q}$ and ${p_{n,k}}, \forall n,k$ can be obtained in a similar way to step 2 of the IS scheme, in which we first rewrite problem \eqref{Problem_sate_active} as problem \eqref{Problem_find_power_considering_interference1}. 
However, problem \eqref{Problem_find_power_considering_interference1} is easier to solve in the ZF scheme because the interference is eliminated in step 1, i.e., ${\rm{Tr}}\left( {{{\bf{G}}_{n,k}}\overline {\bf{V}} } \right)q + \sum_{(m,l) \ne (n,k)} {{\rm{Tr}}\left( {{{\bf{H}}_{m,n,k}}{{{\bf{\overline W}}}_{m,l}}} \right){p_{m,l}}}=0, \forall n,k$ in \eqref{sum_rate_wp} and $\sum_{(m,l)} {{\rm{Tr}}\left( {{{\bf{H}}_{A,m}}{{{\bf{\overline W}}}_{m,l}}} \right){p_{m,l}}}=0$ in \eqref{cons_SINR_wp}.
Therefore, problem \eqref{Problem_find_power_considering_interference1}  is rewritten as:
\begin{subequations}
\label{Problem_find_power_ignoring_interference1}
\begin{align}
\label{Problem_find_power_ignoring_interference_obj}
&\mathop {\max}_{\left\{ {p_{n,k}} \right\}, {q} } \;
\sum_{(n,k)} {{{\log }_2}\left( {1 + {\text{Tr}}\left( {{{\bf{H}}_{n,n,k}}{{{\bf{\overline W}}}_{n,k}}} \right){p_{n,k}}/\overline \sigma _{n,k}^2} \right)}    \\ 
&\;\;\;{\rm{s}}.{\rm{t}}.\;\;\;\;\;
\eqref{cons_interference_power_sate_wp}, \eqref{cons_transmit_power_terrestrial_BS_wp} - \eqref{cons_transmit_power_0}, \nonumber \\ 
\label{cons_SINR_wp_ign}
&\;\;\;\;\;\;\;\;\;\;\;\;\;{\rm{Tr}}\left( {{{\bf{G}}_{A}}{\overline {\bf{V}}}} \right){q} \geq {\overline \beta}  {\overline \sigma }^2_A.
\end{align}
\end{subequations}
Problem \eqref{Problem_find_power_ignoring_interference1} is convex and can be easily solved by CVX \cite{grant2011cvx}.

\subsubsection{Complexity Analysis}


Firstly, we analyze the complexity of solving problem \eqref{Problem_find_power_ignoring_interference1}.
%
Similar to the analysis in Section \ref{Complexity_Analysis_proposed}, problem \eqref{Problem_find_power_ignoring_interference1} can be transformed into a SOCP problem, which is solved by the interior-point method with complexity ${\cal{O}}\left( {(\widetilde K  + 1)}^{3.5} \log \left( 1/\epsilon  \right) \right)$ given a solution accuracy $\epsilon \!> \!0$ \cite{li2020cooperative}.
The complexity of the ZF scheme is mainly determined by step 2, where only one time is required to solve problem \eqref{Problem_find_power_ignoring_interference1}.
Hence, the complexity of the ZF scheme is lower than that of the IS scheme.

\subsection{Maximal Ratio Combining Scheme}
To simplify the design of the normalized beamforming vector in step 1 of the IS and ZF schemes, we propose an MRC scheme based on the idea of maximum ratio combination.
Specifically, in step 1, the normalized beamforming vectors ${\overline {\bf{v}}}$ and ${\overline {\bf{w}} _{n,k}}, \forall n,k$ are obtained by:  
\begin{subequations}
\label{Maximal_Ratio_Combining_Step1}
\begin{align}
&{\overline {\bf{v}}} = {{\bf{g}}_{A}}/{\left\| {{{\bf{g}}_{A}}} \right\|_2},\\
&{\overline {\bf{w}} _{n,k}} = {\bf{h}}_{n,n,k}/{\left\| {{\bf{h}}_{n,n,k}} \right\|_2}, \forall n,k.
\end{align}
\end{subequations}

In step 2, given ${\overline {\bf{v}}}$ and ${\overline {\bf{w}} _{n,k}}, \forall n,k$, ${q}$ and ${p_{n,k}}, \forall n,k$ are obtained by Algorithm \ref{algorithm_power}.
%
Since step 1 in the MRC scheme uses a simple closed-form solution to find ${\overline {\bf{v}}}$ and ${\overline {\bf{w}} _{n,k}}, \forall n,k$, and step 2 is the same as that in the IS scheme, the complexity of the MRC scheme is lower than that of the IS scheme.
In addition, the complexity of step 1 of the MRC scheme is lower than that of step 1 of the ZF scheme, while the complexity of step 2 of the MRC scheme is higher than that of step 2 of the ZF scheme.

\section{Simulation Results}
\label{sec_Simulation_Results}

In this section, simulation results are shown to evaluate the performance of the PIBF scheme and the low-complexity schemes\footnote{Note that the channel realizations for which feasible solutions cannot be obtained by the PIBF scheme or other schemes are not shown.} and demonstrate the advantages of HCSSA compared with TCSSA. 
Fig. \ref{scenario_simulation} shows the projection positions of each device on the ground.
The altitudes of the satellite and the aerial user are $3.5786 \times {10^7}{\rm{m}}$ and $10{\rm{km}}$, respectively. 
The altitudes of the other devices are $0{\rm{m}}$.
\begin{figure}[t]
\centering
\includegraphics[width=6cm]{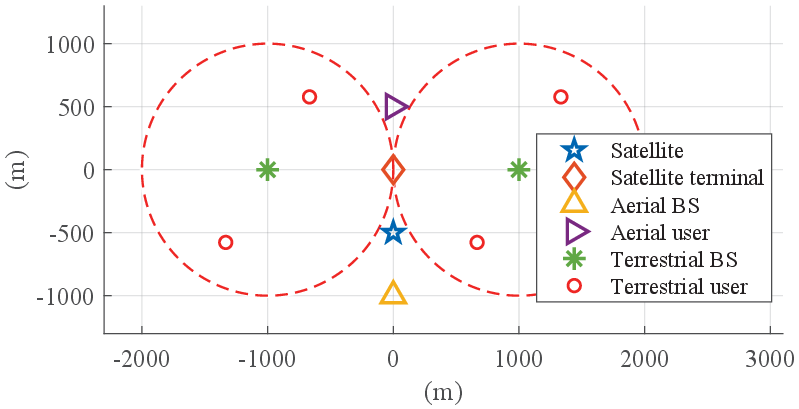}
\caption{The simulation scenario.}
\label{scenario_simulation}
\vspace{-0.4cm}
\end{figure}
The corresponding simulation parameters are given in Table \ref{table_simulation_setting} \cite{abdi2003new, an2016secure,li2018robust, lin2020robust}.
The beamforming vector for the satellite terminal is given by ${\bf{u}} = \sqrt{p_{\rm{s}}} {\bf{f}}_{S}/{\left\| {\bf{f}}_{S} \right\|_2}$, where ${\bf{f}}_{S} \in \mathbb{C}^{M_S \times 1}$ is the channels from the satellite to the satellite terminal, which also adopts the channel model in \eqref{Shadowed_Rician_channel}.
Besides, we set $\overline d/\lambda=1/2$ for simplicity.
The beam angles between the satellite and satellite terminal, the aerial user, as well as terrestrial terminals, are set to $0.01^\circ$, $0.4^\circ$, and $0.8^\circ$, respectively \cite{li2018robust}.
%
%
The noise variance of all receivers is set to $\sigma^2$, i.e., $\sigma _{n,k}^2=\sigma _A^2=\sigma^2, \forall n,k$, which is given by $\sigma ^{2} = \overline\kappa {T}{B}$ \cite{lin2020robust}.
For the PIBF scheme, the parameters of Algorithm \ref{algorithm_penalty_IBF_LYW} are $\varepsilon_1=3\times 10^{-3}$, $\varepsilon_2=10^{-3}$, ${t_{\max}}=20$, $\xi=10^{-5}$, and $\omega=10$.
For the IS scheme, the accuracy parameters of Algorithm \ref{algorithm_interference_suppression} and Algorithm \ref{algorithm_power} are given by $\varepsilon_3=10^{-18}$ and $\varepsilon_4=10^{-2}$, respectively.
\begin{table}[htbp]
\centering  
\caption{Main simulation parameters}  
\label{table_simulation_setting}  
\begin{tabular}{cc}  
\toprule  
Parameter & Value \\  
\midrule  
Satellite channel parameters & 
$(\Omega, b, m) = (0.835, 0.126, 10)$  \\
Transmit power of satellite & ${p_{\rm{s}}} = 40{\rm{W}}$  \\
Maximal antenna gain of satellite & ${b_{\max}} = 52.1{\rm{dB}}$  \\
Rician factor & $\kappa=10$  \\
3-dB angle & ${\varphi _{3{\rm{dB}}}} = 0.4^\circ$  \\
Number of antennas & ${M_S} = 7$, ${M_A} = 8$, ${M_G} = 8$ \\
Carrier frequency & $f = 18{\rm{GHz}}$  \\
Signal bandwidth & $B = 0.5{\rm{MHz}}$  \\
Noise temperature & $T = 300{\rm{K}}$  \\
Boltzmann constant & $\overline\kappa = 1.38 \times 10^{-23}{\rm{J/K}}$\\
\bottomrule  
\end{tabular}
\end{table}

\begin{figure}[t]
\centering
\includegraphics[width=7cm]{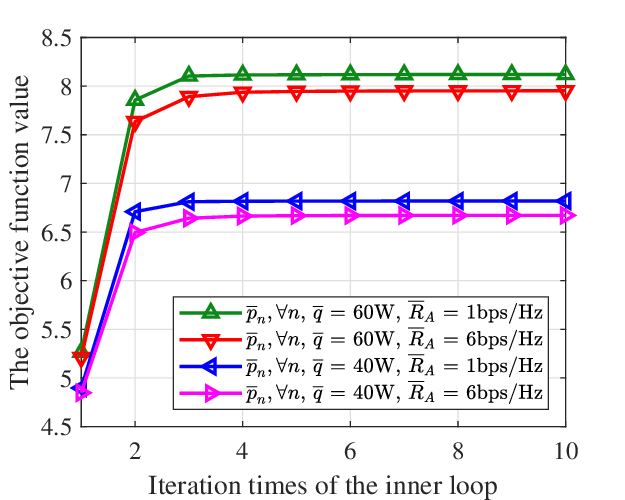}
\caption{The objective function value of problem \eqref{Problem_CVX_penalty} versus the number of inner loop iterations of Algorithm \ref{algorithm_penalty_IBF_LYW}: ${\overline I _S} = 2 \times {10^{ - 12}}{\rm{mW}}$.}
\label{sum_rate_conver_PIBF_in}
\vspace{-0.4cm}
\end{figure}
Fig. \ref{sum_rate_conver_PIBF_in} shows the convergence performance of the inner loop in the first outer loop of Algorithm \ref{algorithm_penalty_IBF_LYW} in one channel realization.
It can be seen that as the number of iterations of the inner loop increases, the objective function value of problem \eqref{Problem_CVX_penalty} increases and converges within 10 inner loop iterations in all cases.
%
%
%
%
\begin{figure}[t]
\centering
\includegraphics[width=7cm]{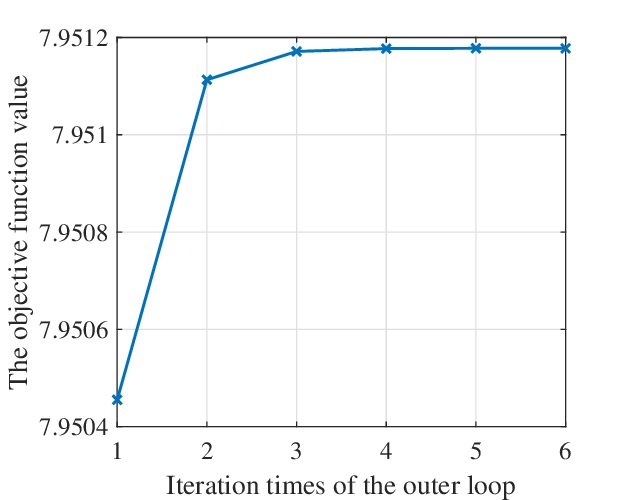}
\caption{The objective function value of problem \eqref{Problem_CVX_penalty} versus the number of outer loop iterations of Algorithm \ref{algorithm_penalty_IBF_LYW}: ${\overline I _S} = 2 \times {10^{ - 12}}{\rm{mW}}$, ${\overline R_A} = 6{\rm{bps/Hz}}$, ${\overline p _0} = 60{\rm{W}}$, and ${\overline p _n} = 60{\rm{W}},\forall n$.}
\label{sum_rate_conver_PIBF_out}
\vspace{-0.4cm}
\end{figure}
Fig. \ref{sum_rate_conver_PIBF_out} shows the convergence performance of the outer loop of Algorithm \ref{algorithm_penalty_IBF_LYW} in one channel realization.
It can be seen that as the number of iterations of the outer loop increases, the objective function value of problem \eqref{Problem_CVX_penalty} increases and converges.

\begin{figure}[t]
\centering
\includegraphics[width=7cm]{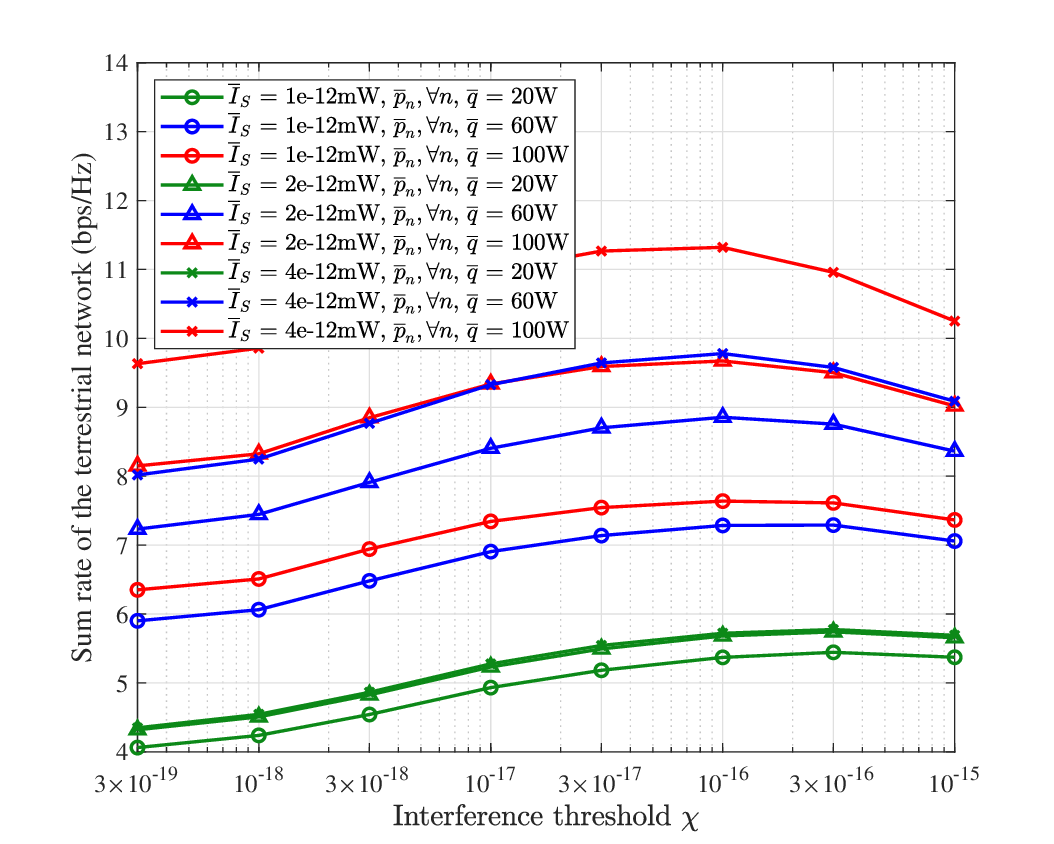}
\caption{Effect of $\chi$ on the sum rate of the terrestrial network (the IS scheme): ${\overline R_A} = 3{\rm{bps/Hz}}$.}
\label{sum_rate_interfer_thre}
\vspace{-0.4cm}
\end{figure}
Fig. \ref{sum_rate_interfer_thre} and Fig. \ref{sum_rate_conver_ISstep2_obj} show the parameter selection and convergence performance of the IS scheme, respectively.
Fig. \ref{sum_rate_interfer_thre} shows the impact of the interference threshold $\chi$ of the IS scheme on the sum rate of the terrestrial network.
In most cases, the sum rate of the terrestrial network can approach the maximum value when $\chi=10^{-16}$.
Therefore, for simplicity, we set $\chi=10^{-16}$ in the subsequent simulations.
\begin{figure}[t]
\centering
\includegraphics[width=7cm]{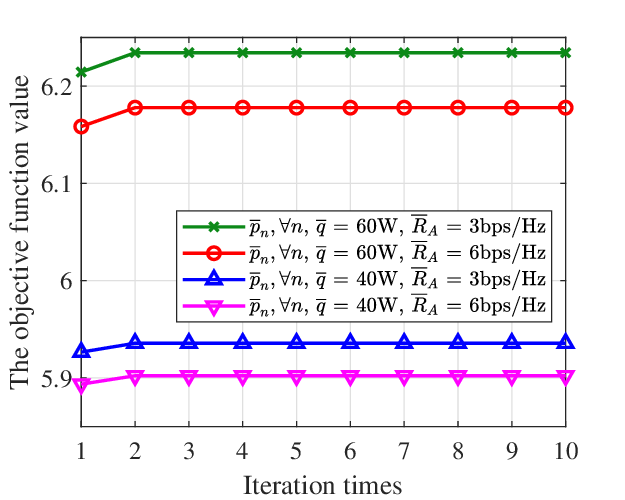}
\caption{The objective function value of problem \eqref{Problem_find_power_considering_interference} versus the number of iterations of Algorithm \ref{algorithm_power}: ${\overline I _S} = 2 \times {10^{ - 12}}{\rm{mW}}$.}
\label{sum_rate_conver_ISstep2_obj}
\vspace{-0.4cm}
\end{figure}
Fig. \ref{sum_rate_conver_ISstep2_obj} shows the convergence performance of Algorithm \ref{algorithm_power} in one channel realization.
It can be seen that as the number of iterations increases, the objective function value of problem \eqref{Problem_find_power_considering_interference} increases and converges within 2 iterations in all cases.

\begin{figure}[t]
\centering
\includegraphics[width=7cm]{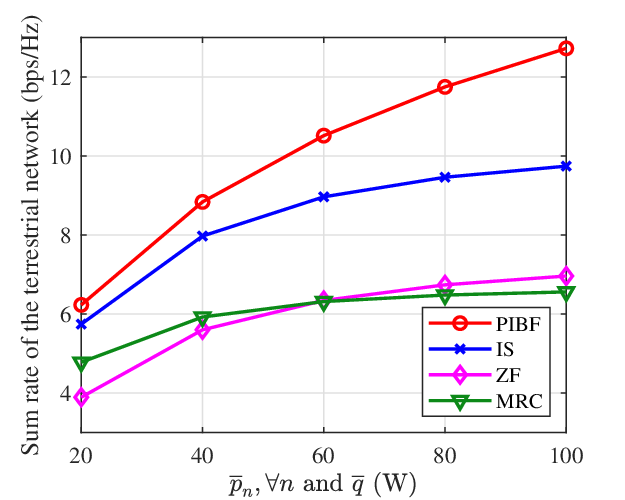}
\caption{Effect of ${\overline p _n},\forall n$ and ${\overline q}$ on the sum rate of terrestrial network: ${\overline I _S} = 2\times {10^{ - 12}}{\rm{mW}}$ and ${\overline R_A} = 3{\rm{bps/Hz}}$.}
\label{sum_rate_max_trans_power}
\vspace{-0.4cm}
\end{figure}

Fig. \ref{sum_rate_max_trans_power} shows the impact of the maximum transmit power of the terrestrial BSs and the aerial BS, i.e., ${\overline p _n},\forall n$ and ${\overline q}$ on the sum rate of terrestrial network.
As ${\overline p _n},\forall n$ and ${\overline q}$ increase, 
the sum rate of the terrestrial network for all schemes in HCSSA increases. 
This is because the terrestrial BSs can transmit signals at higher power.
Hence, a higher sum rate of the terrestrial network is obtained.
Besides, the performance of the PIBF scheme is better than other schemes.
Moreover, when ${\overline p _n},\forall n$ and ${\overline q}$ are high, the ZF scheme outperforms the MRC scheme \cite{yang2018cooperative}.
Intuitively, this is because when the transmit power is high, the interference caused by the beamforming vector obtained by the MRC scheme is severe so the rate of other channels will be significantly reduced.
When the transmit power is low, the MRC scheme can target the intended user well and the impact of interference is small, so a higher sum rate can be obtained.

\begin{figure}[t]
\centering
\includegraphics[width=7cm]{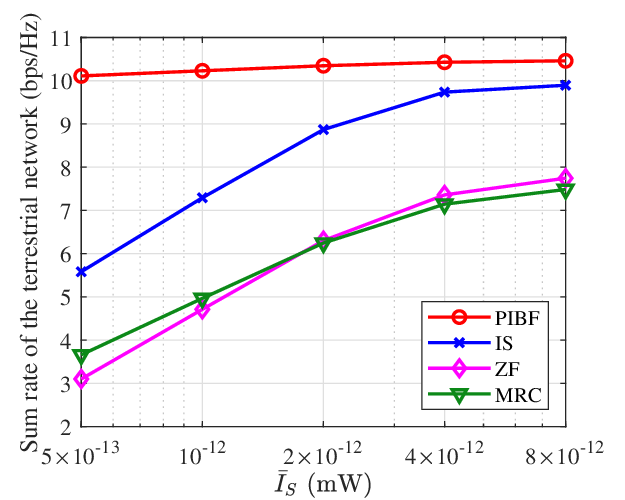}
\caption{Effect of ${\overline I _S}$ on the sum rate of terrestrial network: ${\overline R_A} = 3{\rm{bps/Hz}}$, ${\overline p _0} = 60{\rm{W}}$, and ${\overline p _n} = 60{\rm{W}},\forall n$.}
\label{sum_rate_max_interfer_temp}
\vspace{-0.4cm}
\end{figure}

Fig. \ref{sum_rate_max_interfer_temp} shows the impact of the interference temperature of satellite terminal ${\overline I _S}$ on the sum rate of the terrestrial network.
As ${\overline I _S}$ increases, the sum rate of the terrestrial network of all schemes in HCSSA increases.
This is because an increase in ${\overline I _S}$ allows the terrestrial BSs to transmit signals with higher transmit power for obtaining a higher rate.
Besides, the performance of the PIBF scheme is better than other schemes.
Moreover, when ${\overline I _S}$ are low, the MRC scheme has better performance than the ZF scheme since the low interference temperature leads to limited transmit power.

In order to show the benefits of HCSSA compared to TCSSA, we formulate the problem with TCSSA:
\begin{subequations}
\label{Problem_sate_active_2prior}
\begin{align}
&\mathop {\max}_{\left\{ {{\bf{w}}_{n,k}}\right\}, {\bf{v}} } \;
\sum_{(n,k)} {{{\log }_2}\left( {1 + {\gamma _{n,k}}} \right)}  + {\log _2}\left( {1 + {\beta}} \right)\\
&\;\;\;{\rm{s}}.{\rm{t}}.\;\;\;\;\;\;
\eqref{cons_interference_power_sate}-
\eqref{cons_transmit_power_aerial_BS}, \nonumber
\end{align}
\end{subequations}
where the objective function is the sum of the rates of the terrestrial network and the aerial network.
Problem \eqref{Problem_sate_active_2prior} can be solved by the PIBF scheme and the low-complexity schemes.
The results obtained by solving problem \eqref{Problem_sate_active_2prior} with the corresponding scheme are labeled with "TCSSA".



\begin{figure}[t]
\centering
\includegraphics[width=7cm]{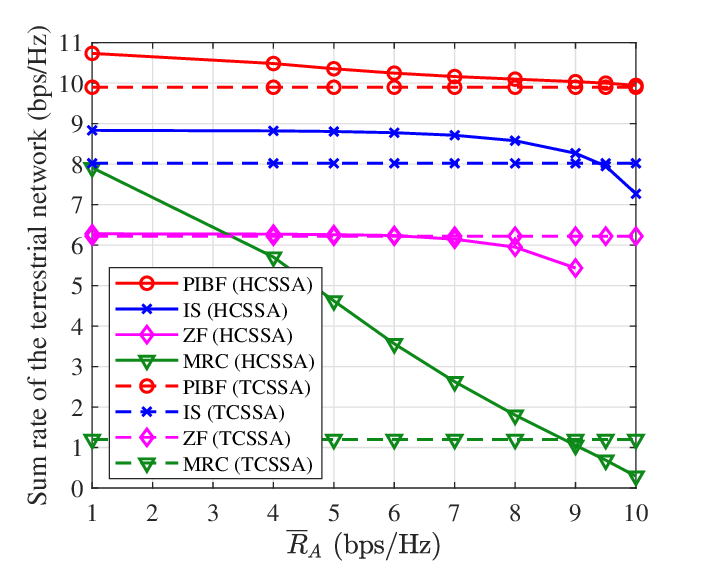}
\caption{Effect of ${\overline R_A}$ on the sum rate of terrestrial network: ${\overline I _S} = 2\times {10^{ - 12}}{\rm{mW}}$, ${\overline p _0} = 60{\rm{W}}$, and ${\overline p _n} = 60{\rm{W}},\forall n$.}
\label{sum_rate_min_SINR_all}
\vspace{-0.4cm}
\end{figure}

Fig. \ref{sum_rate_min_SINR_all} shows the impact of the minimum rate requirement for the aerial user ${\overline R_A}$ on the sum rate of the terrestrial network.
As ${\overline R_A}$ increases, the sum rate of the terrestrial network of all schemes in HCSSA decreases.
This is because the higher rate requirement of the aerial network results in fewer resources available to the terrestrial network.
Moreover, the rates obtained by all schemes in TCSSA remain unchanged since ${\overline R_A}$ is not considered in problem \eqref{Problem_sate_active_2prior}.
Besides, the performance of the PIBF scheme is better than other schemes.
Note that as ${\overline R_A}$ increases, the sum rate of the terrestrial network of the MRC scheme in HCSSA decreases rapidly.
Intuitively, this is due to the fact that the terrestrial BSs cause severe interference to the aerial user in step 1, which can only be mitigated by reducing the transmit power of the terrestrial BSs in step 2.
Thus, the reduction in transmit power rapidly reduces the rate.
In addition, it is noted that when ${\overline R_A}$ is large, the PIBF scheme or other schemes may not be able to obtain feasible solutions, in which case, the secondary networks will not transmit signals.
Since the probability of obtaining a feasible solution for the ZF scheme is too low when ${\overline R_A} > 9{\rm{bps/Hz}}$, it is not shown in Fig. \ref{sum_rate_min_SINR_all}.

\begin{figure}[t]
\centering
\includegraphics[width=7cm]{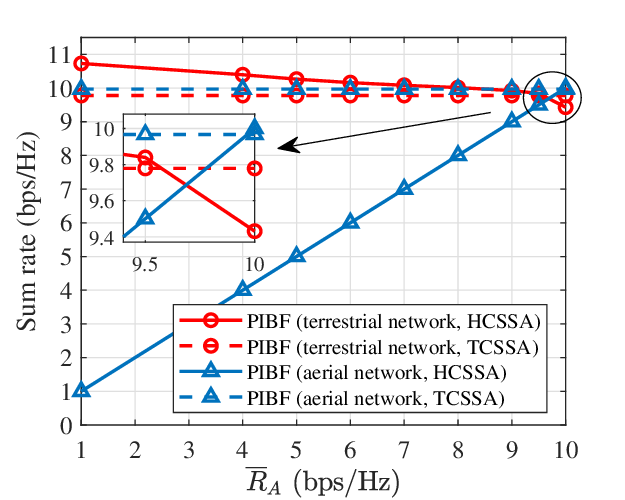}
\caption{Effect of ${\overline R_A}$ on the sum rate of terrestrial network and aerial network respectively: ${\overline I _S} = 2\times{10^{ - 12}}{\rm{mW}}$, ${\overline p _0} = 60{\rm{W}}$, and ${\overline p _n} = 60{\rm{W}},\forall n$.}
\label{sum_rate_min_SINR}
\vspace{-0.4cm}
\end{figure}


Fig. \ref{sum_rate_min_SINR} shows the impact of the minimum rate requirement for the aerial user ${\overline R_A}$ on the sum rate of the aerial network and the terrestrial network, respectively.
It is worth noting that when ${\overline R_A}=10{\rm{bps/Hz}}$, the sum rate of the terrestrial network obtained by the PIBF scheme in HCSSA is lower than that in TCSSA.
This shows that the terrestrial network sacrifices its own rate to ensure the high rate requirement of the aerial network.
%
Compared to TCSSA, the advantages of HCSSA are shown in the following two aspects.
On the one hand, when ${\overline R_A} \leq 9.5{\rm{bps/Hz}}$, the rate requirement of the aerial network is satisfied in HCSSA and TCSSA.
The excess resources of the aerial network are available for the terrestrial network to obtain a higher rate in HCSSA.
However, it is not available in TCSSA.
On the other hand, when ${\overline R_A}$ reaches $10{\rm{bps/Hz}}$, the rate requirement of the aerial network is not satisfied in TCSSA.
However, this requirement is preferentially met in HCSSA.

\section{Conclusions}
\label{sec_Conclusions}
In this paper, a hierarchical cognitive spectrum sharing architecture, i.e., HCSSA, has been proposed for SAGIN by dividing the secondary networks into a preferential one and an ordinary one.
Specifically, a satellite network shares its
spectrum with an aerial network and a terrestrial network if the received interference of the satellite terminal is below a threshold.
Besides, the aerial network has a higher priority than the terrestrial network, and its performance is ensured by a QoS constraint.
%
We have formulated a hierarchical cognitive spectrum sharing problem to maximize the sum rate of the terrestrial network by jointly optimizing the transmit beamforming vectors of the aerial and terrestrial networks while meeting the interference temperature constraint of the satellite network and the rate constraint of the aerial network.
%
To solve this non-convex problem, we have proposed a PIBF scheme by exploiting the penalty method and SCA technique.
Moreover, we also have developed three low-complexity schemes, where the beamforming vectors can be obtained by optimizing the normalized beamforming vectors and power control. 
Simulation results have compared the performance of the PIBF scheme with that of the low-complexity schemes and have illustrated the advantages of HCSSA compared with TCSSA.

\section*{Appendix A}
To prove Theorem \ref{theorem_convergence}, we will  prove the following formula:
\begin{align}
\label{theorem_conver_proof}
&\mu\left( {{\bf{V}}^{(t + 1)}}, {{\bf{W}}^{(t + 1)}} \right)\nonumber \\
&\mathop  \ge ^{(a)} \phi({{\bf{V}}^{(t + 1)}}, {{\bf{W}}^{(t + 1)}}, {u^{(t + 1)}};{{\bf{V}}^{(t)}}, {{\bf{W}}^{(t)}}, {u^{(t)}})\nonumber \\
&\mathop  \ge ^{(b)} \phi({\widetilde {\bf{V}}^{(t)}}, {\widetilde {\bf{W}}^{(t)}}, {{\tilde u}^{(t)}};{{\bf{V}}^{(t)}}, {{\bf{W}}^{(t)}}, {u^{(t)}})\nonumber \\
&\mathop  \ge ^{(c)} \phi({{\bf{V}}^{(t)}}, {{\bf{W}}^{(t)}}, {u^{(t)}};{{\bf{V}}^{(t)}}, {{\bf{W}}^{(t)}}, {u^{(t)}})\nonumber \\
&\mathop  = ^{(d)} \mu\left( {{{\bf{V}}^{(t)}}, {{\bf{W}}^{(t)}}} \right).  
\end{align}

In the following, we will prove (c), (b), (a), and (d) in turn.
Since problem \eqref{Problem_CVX_penalty} is optimized in the $t$ iteration, inequality (c) holds.
After problem \eqref{Problem_CVX_penalty} is optimized in the $t$ iteration, we update the Taylor expansion points for ${\bf{V}}$, ${\bf{W}}_{n,k}, \forall n,k$, and $u_{n,k}, \forall n,k$ and have ${{\bf{V}}^{(t+1)}} = \widetilde {\bf{V}}^{(t)}$, ${\bf{W}}_{n,k}^{(t+1)} = \widetilde {\bf{W}}_{n,k}^{(t)}, \forall n,k$, and $u_{n,k}^{(t+1)} = \ln \left({\alpha_{n,k}}\left( {{\bf{V}}^{(t + 1)}}, {{\bf{W}}^{(t + 1)}} \right)\right)$, i.e., \eqref{update_v_t}.

From \eqref{Problem_CVX_auxi_vari_v_Taylor1}, the Taylor expansion of ${e^{\widetilde u_{n,k}^{(t)}}}$ is given by ${e^{u_{n,k}^{( t )}}}\left( \widetilde u_{n,k}^{(t)} - u_{n,k}^{( t )} + 1 \right) \le {e^{\widetilde u_{n,k}^{(t)}}}, \forall n,k$.
Besides, when the optimal solution of problem \eqref{Problem_CVX_penalty} is obtained, the equalities in constraint \eqref{Problem_CVX_auxi_vari_v_Taylor} hold, i.e., ${e^{u_{n,k}^{( t )}}}\left( {\widetilde u_{n,k}^{(t)} - u_{n,k}^{( t )} + 1} \right) = {\alpha_{n,k}}\left( {{\widetilde {\bf{V}}}^{(t)}}, {{\widetilde {\bf{W}}}^{(t)}} \right)$.
Thus, we have $\ln \left( {{\alpha _{n,k}}\left( {{{\widetilde {\bf{V}}}^{(t)}},{{\widetilde {\bf{W}}}^{(t)}}} \right)} \right) \le \tilde u_{n,k}^{(t)}$.
Besides, from \eqref{update_v_t}, we have $u_{n,k}^{(t + 1)} = \ln \left( {{\alpha _{n,k}}\left( {{{\widetilde {\bf{V}}}^{(t)}},{{\widetilde {\bf{W}}}^{(t)}}} \right)} \right)$.
Thus, we have $u_{n,k}^{(t + 1)} \le \tilde u_{n,k}^{(t)}$ and obtain inequality (b) in \eqref{theorem_conver_proof}.

Since ${\overline F}^{(t)}$ is the convex upper bound of the penalty term $F$, we have:
\begin{align}
\label{rank1_penalty_term}
\overline F\left( {{{\widetilde {\bf{V}}}^{(t)}}, {{\widetilde {\bf{W}}}^{(t)}};{{\bf{V}}^{(t)}}, {{\bf{W}}^{(t)}}} \right) \ge F\left( {{{\widetilde {\bf{V}}}^{(t)}}, {{\widetilde {\bf{W}}}^{(t)}}} \right),
\end{align}
where the equality holds when $\widetilde {\bf{W}}_{n,k}^{(t)} = {\bf{W}}_{n,k}^{(t)},\forall n,k$ and ${\widetilde {\bf{V}}^{(t)}} = {{\bf{V}}^{(t)}}$.
From \eqref{rank1_penalty_term} and \eqref{update_v_t}, inequality (a) holds in \eqref{theorem_conver_proof}.
From $\overline F\left( {{{\bf{V}}^{(t)}}, {{\bf{W}}^{(t)}};{{\bf{V}}^{(t)}}, {{\bf{W}}^{(t)}}} \right) = F\left( {{{\bf{V}}^{(t)}}, {{\bf{W}}^{(t)}}} \right)$ and \eqref{update_v_t}, equality (d) holds in \eqref{theorem_conver_proof}.

\small
\bibliographystyle{IEEEtran}
\bibliography{IEEEabrv,CognitiveSAGIN}
\end{document}